%% file: master-arxiv.tex
\documentclass[11pt,a4paper,reqno]{article}
\pdfoutput=1

\input{./arXiv/packages-core}
\input{./arXiv/arxiv-packages}
\input{./arXiv/arxiv-macros}

\author[1]{Ali Arslan\thanks{Corresponding author: \href{mailto:a.arslan18@imperial.ac.uk}{a.arslan18@imperial.ac.uk}}}
\author[1]{Giovanni Fantuzzi}
\author[2]{John Craske}
\author[1]{Andrew Wynn}
\affil[1]{\small Dept. of Aeronautics, Imperial College London, London, SW7 2AZ, UK}
\affil[2]{\small Dept. of Civil and Environmental Engineering, Imperial College London, SW7 2AZ, UK}
\date{\today}

\title{\Large\bf 
Rigorous scaling laws for internally heated convection at infinite Prandtl number  }

\begin{document}


\maketitle

\begin{abstract}\noindent
    New bounds are proven on the mean vertical convective heat transport, $\wT$, for uniform internally heated (IH) convection in the limit of infinite Prandtl number. 
    For fluid in a horizontally-periodic layer between isothermal boundaries, we show that $\wT \leq \frac12 - c R^{-2}$, where $R$ is a nondimensional `flux' Rayleigh number quantifying the strength of internal heating and  $c = 216$.
    Then, $\wT = 0$ corresponds to vertical heat transport by conduction alone, while $\wT > 0$ represents the enhancement of vertical heat transport upwards due to convective motion.
    If, instead, the lower boundary is a thermal insulator, then we obtain $\wT \leq \frac12 - c R^{-4}$, with $c\approx 0.0107$. This result implies that the Nusselt number \Nu, defined as the ratio of the total-to-conductive heat transport, satisfies $\Nu \lesssim \Ra^{4}$.
    Both bounds are obtained by combining the background method with a minimum principle for the fluid's temperature and with Hardy--Rellich inequalities to exploit the link between the vertical velocity and temperature.
    In both cases, power-law dependence on \Ra\ improves the previously best-known bounds, which, although valid at both infinite and finite Prandtl numbers, approach the uniform bound exponentially with \Ra.
\end{abstract}

\input{./arXiv/arxiv-v6}

\input{arXiv/arxiv-appendix}

\input{./master-arxiv.bbl}

\end{document}

%% file: arXiv/packages-core.tex
\usepackage[margin=3.4cm]{geometry}

\usepackage[noblocks]{authblk}

\usepackage{amsmath}
\usepackage{amsfonts}
\usepackage{amsthm}
\usepackage{amssymb}
\usepackage{bbm}
\numberwithin{equation}{section}

\usepackage{setspace}

\usepackage{titlesec}
\titleformat{\section}{\Large\bfseries}{\thesection}{1.5ex}{}
\titleformat{\subsection}{\large\bfseries}{\thesubsection}{1.5ex}{}

\usepackage{graphicx}
\usepackage{booktabs}
\usepackage[suffix=]{epstopdf}
\usepackage[font=footnotesize,labelfont=bf]{caption}
\usepackage[colorlinks]{hyperref}
\hypersetup{citecolor=blue,urlcolor=blue,linkcolor=blue}
\newtheorem{theorem}{Theorem}
\newtheorem{proposition}{Proposition}
\newtheorem{lemma}{Lemma}
\newtheorem{corollary}{Corollary}[theorem]
\newtheorem*{conjecture*}{Conjecture}
\theoremstyle{definition}
\newtheorem{definition}{Definition}

\theoremstyle{remark}
\newtheorem{remark}{Remark}

\usepackage[sort&compress,numbers]{natbib}

\usepackage{cleveref}
\crefname{equation}{}{}
\crefname{enumi}{}{}
\crefname{theorem}{Theorem}{Theorems}
\crefname{corollary}{Corollary}{Corollaries}
\crefname{example}{Example}{Examples}
\crefname{lemma}{Lemma}{Lemmas}
\crefname{proposition}{Proposition}{Propositions}
\crefname{figure}{Figure}{Figures}
\crefname{table}{Table}{Tables}
\crefformat{section}{\S#2#1#3}
\crefmultiformat{section}{\S#2#1#3}{ and~\S#2#1#3}{, \S#2#1#3}{, and~\S#2#1#3}
\crefformat{subsection}{\S#2#1#3}
\crefmultiformat{subsection}{\S#2#1#3}{ and~\S#2#1#3}{, \S#2#1#3}{, and~\S#2#1#3}
\crefname{appendix}{appendix}{appendices}
\Crefname{equation}{}{}
\Crefname{enumi}{}{}
\Crefname{theorem}{Theorem}{Theorems}
\Crefname{corollary}{Corollary}{Corollaries}
\Crefname{example}{Example}{Examples}
\Crefname{lemma}{Lemma}{Lemma}
\Crefname{proposition}{Proposition}{Proposition}
\Crefname{figure}{Figure}{Figures}
\Crefname{table}{Table}{Tables}
\Crefformat{section}{Section~#2#1#3}
\Crefmultiformat{section}{Sections~#2#1#3}{ and~#2#1#3}{, #2#1#3}{, and~#2#1#3}
\Crefformat{subsection}{Section~#2#1#3}
\Crefmultiformat{subsection}{Sections~#2#1#3}{ and~#2#1#3}{, #2#1#3}{, and~#2#1#3}
\Crefname{appendix}{Appendix}{Appendices}

%% file: arXiv/arxiv-packages.tex
\usepackage{mathtools}
\usepackage{mathrsfs}
\usepackage{esint}
\usepackage{xcolor}
\usepackage{tikz}
\usepackage[shortlabels]{enumitem}

%% file: arXiv/arxiv-macros.tex
\newcommand{\uvec}[1]{\boldsymbol{\hat{\mathbf{#1}}}}
\newcommand\Ra{\mbox{\textit{R}}}  
\renewcommand\Pr{\mbox{\textit{Pr}}}  
\newcommand\Nu{\mbox{\textit{Nu}}}  
\newcommand{\wT}{\mbox{$\smash{\mean{wT}}$}}
\newcommand{\abs}[1]{\left\vert #1 \right\vert}

\newcommand{\vx}{\boldsymbol{x}}
\newcommand{\vu}{\boldsymbol{u}}

\newcommand{\volav}[1]{\left\langle #1 \right\rangle}
\newcommand{\smallvolav}[1]{\langle #1 \rangle}

\newcommand{\timeav}[1]{\overline{#1}}
\newcommand{\mean}[1]{\timeav{\volav{#1}}}
\newcommand{\bR}{\mathbb{R}}
\newcommand{\Tspace}{\mathcal{H}}
\newcommand{\dz}{{\rm d}z}
\newcommand{\bfield}{\tau} 
\newcommand{\lm}{\lambda} 
\newcommand{\bp}{\beta} 
\newcommand{\translation}[1]{\mathscr{S}_{#1}}
\newcommand{\subeqref}[2]{\hyperref[#1]{(\ref*{#1}#2)}}

\definecolor{matlabblue}{RGB}{0,113,188}
\definecolor{matlabred}{RGB}{216,82,24}
\definecolor{mygrey}{rgb}{0.7,0.7,0.7}
\definecolor{matlabgreen}{rgb}{0.4660,0.6740,0.1880} 
\definecolor{lightblue}{RGB}{0,180,200}
\definecolor{lightred}{RGB}{250,100,40}

\definecolor{matlabyellow}{rgb}{0.93,0.69,0.13} 
\definecolor{colorbar1}{rgb}{1.000000,0.909091,0.000000}
\definecolor{colorbar2}{rgb}{1.000000,0.818182,0.000000}
\definecolor{colorbar3}{rgb}{1.000000,0.727273,0.000000}
\definecolor{colorbar4}{rgb}{1.000000,0.636364,0.000000}
\definecolor{colorbar5}{rgb}{1.000000,0.545455,0.000000}
\definecolor{colorbar6}{rgb}{1.000000,0.454545,0.000000}
\definecolor{colorbar7}{rgb}{1.000000,0.363636,0.000000}
\definecolor{colorbar8}{rgb}{1.000000,0.272727,0.000000}
\definecolor{colorbar9}{rgb}{1.000000,0.181818,0.000000}
\definecolor{colorbar10}{rgb}{1.000000,0.090909,0.000000}
\definecolor{colorbar11}{rgb}{1.000000,0.000000,0.000000}
\definecolor{colorbar12}{rgb}{0.909091,0.000000,0.000000}
\definecolor{colorbar13}{rgb}{0.818182,0.000000,0.000000}
\definecolor{colorbar14}{rgb}{0.727273,0.000000,0.000000}
\definecolor{colorbar15}{rgb}{0.636364,0.000000,0.000000}
\definecolor{colorbar16}{rgb}{0.545455,0.000000,0.000000}
\definecolor{colorbar17}{rgb}{0.454545,0.000000,0.000000}
\definecolor{colorbar18}{rgb}{0.363636,0.000000,0.000000}
\definecolor{colorbar19}{rgb}{0.272727,0.000000,0.000000}
\definecolor{colorbar20}{rgb}{0.181818,0.000000,0.000000}
\definecolor{colorbar21}{rgb}{0.090909,0.000000,0.000000}
\definecolor{grey}{rgb}{0.7,0.7,0.7}

\newcommand\solidrule[1][10pt]{\rule[0.5ex]{#1}{1.5pt}}
\newcommand\dashedrule{\mbox{\solidrule[2pt]\hspace{2pt}\solidrule[2pt]\hspace{2pt}\solidrule[2pt]}}

\newcommand\dottedrule{\mbox{\solidrule[1pt]\hspace{1pt}\solidrule[1pt]\hspace{1pt}\solidrule[1pt]\hspace{1pt}\solidrule[1pt]\hspace{1pt}\solidrule[1pt]\hspace{1pt}}}


%% file: arXiv/arxiv-v6.tex
\section{Introduction}\label{sec:intro}

Convective flows driven by internal sources of heat have attracted renewed interest in recent years~\cite{arslan2021ih2, arslan2021IH1, Wang2020, creyssels_2021, kumar2021ihc, Tobasco2022}. Such flows are commonly encountered in geophysics, where atmospheric convection~\cite{emanuel1994atmospheric} and mantle convection~\cite{schubert2001mantle, MantleConvectioninTerrestrialPlanets} are typical examples. They also exhibit unique features not seen in boundary-driven Rayleigh--B\'enard convection: for instance, it has recently been observed experimentally that internally heated (IH) convection can transport heat more efficiently than Rayleigh-B\'enard convection~\cite{miquel2019convection,bouillaut2019transition}. Nevertheless, the former remains much less studied. In particular, it remains a largely open challenge to rigorously predict how key statistical properties such as the mean vertical heat flux depend on the heating strength and on the fluid's Prandtl number $\Pr$, defined as the ratio between the fluid's kinematic viscosity $\nu$ and its thermal diffusivity $\kappa$.

One source of difficulty for mathematical studies of IH convection is that the mean thermal dissipation, the mean viscous dissipation, and the mean vertical convective heat flux cannot all be related to each other via \textit{a priori} relationships. This is in contrast with Rayleigh-B\'enard convection, where such relationships enable one to rigorously bound the convective heat transfer through variational analysis of the mean thermal dissipation~\cite{doering1996variational}. Applying the same strategy to IH flows yields bounds on the mean temperature of the fluid~\cite{lu2004bounds,whitehead2011internal,goluskin2016internally,whitehead2012slippery} but not the convective heat flux. Recently, this variational strategy was extended by taking into account a minimum principle for the temperature, leading to bounds on the mean convective heat flux that approach a constant exponentially fast as the heating strength is increased~\cite{arslan2021IH1,kumar2021ihc}. Here, we demonstrate that these bounds can be improved to algebraic powers when the Prandtl number is taken to be infinite. 

Using standard non-dimensional variables~\cite{goluskin2016internally}, we consider a fluid in a horizontally periodic domain $\Omega = \mathbbm{T}_{[0,L_x]} \times \mathbbm{T}_{[0,L_y]}   \times [0,1]$, the motion of which is governed by the infinite Prandtl number Boussinessq equations
\begin{subequations}\label{e:governing-equations}
    \begin{align}
    \nabla \cdot \vu &= 0\, , \label{continuit} \\
    \nabla p &= \Delta\vu + \Ra\, T \uvec{z}\, , \label{nondim_momentum} \\
    \partial_t T + \vu\cdot \nabla T  &= \Delta T + 1.
    \label{nondim_energy}
\end{align}
\end{subequations}
Here, $\vu=(u,v,w)$ is the fluid velocity in cartesian components, $p$ is the pressure, $T$ is the temperature, and the unit forcing in~\eqref{nondim_energy} represents the non-dimensional internal heating rate. The flow is controlled by a `flux' Rayleigh number that measures the destabilising effect of the heating compared to the stabilising effects of diffusion,
\begin{equation}
    \Ra := \frac{g \alpha Q d^{5}}{\rho c_p \nu \kappa^{2}}.
\end{equation}
Here $g$ is the acceleration of gravity, $\rho$ is the density, $c_p$ is the specific heat capacity, $\alpha$ is the thermal expansion coefficient and $Q$ is the heating rate per unit volume.

We consider two separate configurations that differ in the choice of boundary conditions at the top ($z=1$) and bottom ($z=0$) of the domain. In the first configuration, referred to as IH1 and sketched in \cref{fig:config}\textit{(a)}, the velocity satisfies no-slip conditions and the temperature of both vertical boundaries is held at a constant value, which can be taken as zero without loss of generality. Hence, we enforce
\begin{subequations}
\begin{equation}
\text{IH1:} \qquad  \vu\vert_{z\in\{0,1\}}=\boldsymbol{0}, \qquad
T\vert_{z=0} = 0, \qquad
T\vert_{z=1} = 0.
\label{bc_T_IH1}
\end{equation}
In the second configuration, illustrated by \cref{fig:config}\textit{(b)}, the bottom plate is replaced by a perfect thermal insulator, giving
\begin{equation}
\text{IH3:} \qquad \vu\vert_{z\in\{0,1\}}=\boldsymbol{0}, \qquad
\partial_z T\vert_{z=0} =0, \qquad
T\vert_{z=1} = 0.   
\label{bc_T_IH3} 
\end{equation}
\end{subequations}

We seek bounds on the mean vertical convective heat transport, \wT, a quantity that is directly proportional to the viscous dissipation for both the IH1 and IH3 configurations. 
Throughout this paper overbars denote infinite time averages, while angled brackets denote volume averages:
\begin{subequations}
\begin{align}
    \volav{f} &= 
    \frac{1}{L_x L_y} \int^{L_x}_0 \int^{L_y}_0 \int^{1}_{0} f(x,y,z,t) \textrm{d}z\, \textrm{d}y\, \textrm{d}x,\\
    \mean{f} &= \limsup_{\tau\rightarrow \infty} \frac{1}{\tau} 
    \int^{\tau}_{0} \volav{f} \textrm{d}t.
\end{align}
\end{subequations}

For either set of boundary conditions considered in this paper, it is known that $0 \leq \wT \leq \frac12$ uniformly in \Ra\ and \Pr~\cite{goluskin2012convection}. The zero lower bound is saturated by the (possibly unstable) state in which the flow does not move and heat is transported vertically by conduction alone. The upper bound of $\frac12$, instead, takes on a different meaning depending on the thermal boundary conditions. As will become apparent from \eqref{e:fluxes_ih1} below, for the IH1 configuration a flow with $\wT=\frac12$ would see all heat escape the domain through the upper boundary.  
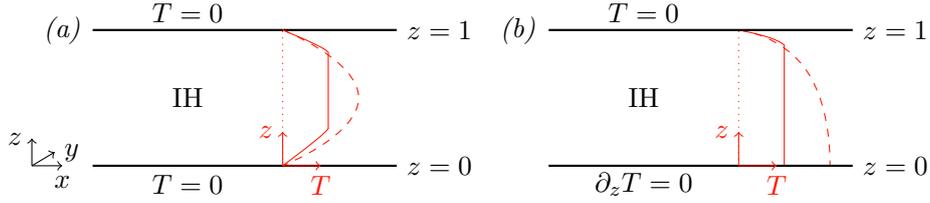
\begin{figure}
    \centering
    \begin{tikzpicture}[every node/.style={scale=0.95}]
    \draw[black,thick] (-6,0) -- (-2,0) node [anchor=west] {$z=0$};
    \draw[black,thick] (-6,1.8) -- (-2,1.8) node [anchor = west] {$z=1$};
    \draw [dashed,colorbar10] plot [smooth, tension = 1] coordinates {(2.5,1.8) (3.4,1.25) (3.7,0)};
    \draw [colorbar10] plot [smooth, tension =1] coordinates {(2.5,1.8) (2.9,1.7) (3.1,1.6)};
    \draw [colorbar10] (3.1,1.6) -- (3.1,0);
    \node at (1.25,-0.25) {$ \partial_z T = 0 $};
    \node at (1.25,2.025) {$ T = 0 $};
    \node at (-4.75,0.9) { IH };
    \draw[->,colorbar10] (-3.5,0) -- (-3.5,0.45) node [anchor = east]{$z$};
    \draw [dotted,colorbar10] (2.5,0) -- (2.5,1.8);
    \draw[->,colorbar10] (-3.5,0) -- (-3,0) node [anchor=north]{$T$};
    \draw[->] (-6.8,0) -- (-6.8,0.36) node [anchor=east]{$z$};
    \draw[->] (-6.8,0) -- (-6.4,0) node [anchor=north]{$x$};
    \draw[->] (-6.8,0) -- (-6.5,0.18) node [anchor=west]{$y$};
    \draw[black,thick] (0,0) -- (4,0) node [anchor=west] {$z=0$};
    \draw[black,thick] (0,1.8) -- (4,1.8) node [anchor = west] {$z=1$};
    \node at (-4.75,-0.25) {$ T = 0 $};
    \node at (-4.75,2.025) {$  T = 0 $};
    \draw [dashed,colorbar10] plot [smooth, tension = 1] coordinates { (-3.5,1.8) (-2.5,0.9) (-3.5,0) };
    \draw [colorbar10] plot [smooth,tension=1] coordinates {(-3.5,1.8) (-3.05,1.6) (-2.9,1.5)};
    \draw [colorbar10] (-2.9,1.5) -- (-2.9,0.5);
    \draw [colorbar10] plot [smooth, tension =1] coordinates {(-2.9,0.5) (-3.05,0.35) (-3.5,0)};
    \draw [dotted,colorbar10] (-3.5,0) -- (-3.5,1.8) ;
    \draw[->,colorbar10] (2.5,0) -- (2.5,0.45) node [anchor = east]{$z$};
    \draw[->,colorbar10] (2.5,0) -- (3,0) node [anchor=north]{$T$};
    \node at (1.25,0.9) {\textrm{IH}};
    \node at (-6.4,1.8) {\textit{(a)}};
    \node at (-0.4,1.8) {\textit{(b)}};
    \end{tikzpicture}
    \caption{IH convection with \textit{(a)} isothermal boundaries \eqref{bc_T_IH1}, and \textit{(b)} insulating lower boundary \eqref{bc_T_IH3}. In both panels, IH represents the uniform unit internal heat generation. Red lines denote the conductive temperature profiles ({\color{colorbar10}\dashedrule}) and indicative mean temperature profiles in the turbulent regime ({\color{colorbar10}\solidrule}).}
    \label{fig:config}
\end{figure}
This, however, cannot be achieved at any finite value of \Ra. Precisely, our first main results reveal that the mean vertical heat flux \wT\ is strictly smaller than $\frac12$ by an amount that cannot decrease faster than quadratically as \Ra\ is raised.

\begin{theorem}[Isothermal boundaries, $\Pr=\infty$]\label{thm:ih1-main-result}
Suppose that $\boldsymbol{u}=(u,v,w)$ and $T$ solve \eqref{e:governing-equations} subject to the no-slip isothermal boundary conditions \eqref{bc_T_IH1}. There exists a constant $c>0$ such that, for all sufficiently large $R>0$, 
\begin{equation}\label{e:thm1_b}
    \mean{wT} \leq\frac12 - c\,R^{-2}.
\end{equation}
\end{theorem}%
\begin{remark}\label{rem:prefactors-ih1}
It is shown in \cref{ss:proof-IH1} that \eqref{e:thm1_b} holds with $c = 216$ for any $R > 1892$.
\end{remark}

The bound on \wT\ can be given a clear physical interpretation as a measure of the asymmetry of the heat transport due to heating. Indeed, upon computing $\mean{z\cdot \eqref{nondim_energy} }$ one can show that the average nondimensional heat fluxes through the top and bottom boundaries, denoted by $\mathcal{F}_T$ and $\mathcal{F}_B$ respectively, can be expressed as  
\begin{subequations}\label{e:fluxes_ih1}
\begin{align}
\label{e:top_flux_ih1}
\mathcal{F}_T &:= -\partial_z\mean{T}_h\vert_{z=1} = \frac12 + \mean{wT} , \\
\label{e:bottom_flux_ih1}
\mathcal{F}_B &:= \phantom{-}\partial_z\mean{T}_h\vert_{z=0} = \frac12 - \mean{wT} ,
\end{align}
\end{subequations}
where $\volav{\cdot}_h$ denotes a spatial average over the horizontal directions alone.
Thus, our upper bound on \wT\ immediately implies bounds on $\mathcal{F}_T$ and $\mathcal{F}_B$.

\begin{corollary}
For all sufficiently large $R$,
\begin{equation}\label{e:main-result-ih1}
    \mathcal{F}_T \leq 1 -  c\, R^{-2}
    \qquad \textrm{and} \qquad \mathcal{F}_B \geq c\, R^{-2} .
\end{equation}
\end{corollary}

The results discussed so far apply to the IH1 configuration, where the two horizontal domain boundaries are isothermal. Similar results hold also for the IH3 case, where the bottom boundary is perfectly insulating. In this case, however, the deviation of \wT\ from $\frac12$ may decay as fast as $\Ra^{-4}$.

\begin{theorem}[Insulating bottom \& isothermal top, $\Pr=\infty$] Suppose that $\boldsymbol{u}=(u,v,w)$ and $T$ solve \eqref{e:governing-equations} subject to the boundary conditions \eqref{bc_T_IH3}. There exists a constant $c>0$ such that, for all sufficiently large $\Ra>0$,
\begin{equation}\label{e:thm2_b}
    \mean{wT} \leq \frac12 - c \, \Ra^{-4}.
\end{equation}
\label{thm:ih3}
\end{theorem}%
\begin{remark}\label{rem:prefactors-ih3}
It is shown in \cref{ss:bm-IH3} that \eqref{e:thm2_b} holds with $c \approx 0.0107$ for all $R > 2961$. 
\end{remark}
Since the IH3 boundary conditions imply that the conductive heat flux is positive, the effects of convection on the enhancement of heat transport in the system can be described using a Nusselt number, $\Nu$. This is defined as the ratio of the mean total heat flux to the mean conductive heat flux, and can be expressed in terms of \wT\ as
\begin{equation}
    \Nu = \left( 1-2\mean{wT} \right)^{-1}.
    \label{e:Nu_ih3}
\end{equation}
Thus, our upper bound on \wT\ can be transformed into an upper bound on \Nu.
\begin{corollary}
For all sufficiently large \Ra, $\Nu \leq c R^4$.
\end{corollary}

The proofs of \cref{thm:ih1-main-result,thm:ih3} rely on two key ingredients. The first is a variational problem giving an upper bound on \wT. 
This variational problem is derived by enforcing a minimum principle for the fluid's temperature within the classical ``background method''~\cite{doering1996variational,doering1994variational,constantin1995variational}, which for simplicity we formulate using the language of a more general framework for bounding infinite-time averages~\cite{chernyshenko2014polynomial,Chernyshenko2022,tobasco2018optimal,rosa2020optimal} (see~\cite{Chernyshenko2022,Fantuzzi2022} for further discussion of the link between the two approaches). Using this minimum principle is essential to obtain bounds on \wT\ that asymptote to $\frac12$ from below. This has already been shown for the IH1 configuration at finite \Pr: for this case, without the miniumum principle one obtains only $\wT \lesssim R^{1/5}$~\cite{arslan2021IH1}, while with it one can prove that $\wT \leq \frac12 - O(R^{1/5} \exp{(- R^{3/5})})$ uniformly in \Pr~\cite{kumar2021ihc}. A similar (but not identical) exponentially-varying bound of $\wT \leq \frac12 -  O(R^{-1/5} \exp{(- R^{3/5}))}$ uniformly in \Pr\ was also obtained for the IH3 configuration~\cite{kumar2021ihc}.

The second key ingredient in our proofs are estimates of Hardy--Rellich type, obtained by observing that the reduced momentum equation \eqref{nondim_momentum} determines the vertical velocity field as a function of the temperature field. Specifically, taking the vertical component of the double curl of~\eqref{nondim_momentum} gives
\begin{equation}
\label{e:w_and_T_eq}
\Delta^2 w = -R \, \Delta_h T,
\end{equation}
where $\Delta_h := \partial^2_x + \partial^2_y $, is the horizontal Laplacian. Using the no-slip boundary conditions with the incompressibility condition ~\eqref{continuit}, the vertical velocity $w$ satisfies
\begin{equation}\label{e:w-bcs}
    w\vert_{z=0} = \partial_z w\vert_{z=0} =
    w\vert_{z=1} = \partial_z w\vert_{z=1} = 0.
\end{equation}
Equation \eqref{e:w_and_T_eq} was exploited in Rayleigh-B\'enard convection to improve the scaling of upper bounds on \Nu~\cite{doering2006bounds}. This was achieved by using \eqref{e:w_and_T_eq} to derive inequalities of the Hardy--Rellich type (see \cref{lemma:doering_estimate} below) that help the construction of a background field with a logarithmically-varying stable stratification in the bulk~\cite{doering2006bounds,Whitehead2014mixed}.
Here, we use the same inequalities to construct (different) background fields suited to IH convection, which will enable us to bound \wT\ in the infinite \Pr\ limit.

\section{Bounds for the IH1 configuration}\label{sec:IH1}

We first consider the IH1 configuration, where the top and bottom plate are held at zero temperature. In \cref{ss:bm-IH1}, we show that $\smash{\mean{wT}}$ can be bounded from above by constructing suitably constrained functions of the vertical coordinate $z$. \Cref{ss:ansatz-IH1} describes parametric ans\"atze for such functions, while \cref{ss:lemmas-IH1} establishes auxiliary results that simplify the verification of the constraints and the evaluation of the bound. We then prove~\cref{thm:ih1-main-result} in \cref{ss:proof-IH1} by prescribing \Ra-dependent values of the free parameters in our ans\"atze.

To simplify the notation, we introduce two sets of temperature fields that encode the thermal boundary conditions and the pointwise nonnegativity constraint implied by the minimum principle:
\begin{subequations}
	\begin{gather}
	\label{e:T-space-ih1}
	\Tspace_{n} := \{T \in H^1(\Omega):\; T \text{ is horizontally periodic \&  \eqref{bc_T_IH1}}  \veebar  \eqref{bc_T_IH3} \},\\
	\Tspace_+ := \{ T \in \Tspace_n: T(\vx) \geq 0 \text{ a.e. } \vx \in \Omega \},
	\label{e:T-positive-cone-ih1}
	\end{gather}
\end{subequations}
where $n =$ 1 or 3, depending on the boundary conditions of IH1 and IH3 respectively. In \cref{sec:IH1}, $T$ belongs to $\Tspace_{1}$.

\subsection{Bounding framework}\label{ss:bm-IH1}

To bound $\smash{\mean{wT}}$ we employ the auxiliary function method~\cite{chernyshenko2014polynomial,Fantuzzi2022}. The method relies on the observation that the time derivative of any bounded functional $\mathcal{V}\{T(t)\}$ along solutions of the Boussinesq equations~\eqref{e:governing-equations} averages to zero over infinite time, so 
\begin{equation}\label{e:af-method}
\mean{wT} = \timeav{\volav{wT} + \tfrac{\rm d}{ {\rm d} t}\mathcal{V}\{T(t)\}  }.
\end{equation}
If $\mathcal{V}$ is chosen such that the quantity being averaged on the right-hand side is bounded above pointwise in time, then this pointwise bound is also an upper bound on $\smash{\mean{wT}}$.

Following analysis at finite Prandtl number~\cite{arslan2021IH1}, we restrict our attention to quadratic functionals taking the form
\begin{equation}\label{e:V-IH1}
\mathcal{V}\{T\} = \volav{ \tfrac{\bp}{2} \abs{T}^2 - [\bfield(z)+z-1] T },
\end{equation}
which are parametrized by a positive constant $\bp \in \bR_+$ and a piecewise-differentiable function $\bfield:[0,1] \to \bR$ with square-integrable derivative. We require $\bfield$ to satisfy
\begin{equation}\label{e:psi-bc-ih1}
\bfield(0)=1, \qquad \bfield(1)=0,
\end{equation}
so the coefficient multiplying $T$ in~\eqref{e:V-IH1} vanishes at $z=0$ and $z=1$. This choice enables us to integrate by parts without picking up boundary terms when calculating $\tfrac{\rm d}{ {\rm d} t}\mathcal{V}\{T(t)\}$. To ensure that the resulting expression $\volav{wT} + \tfrac{\rm d}{ {\rm d} t}\mathcal{V}\{T\}$ can be bounded from above poinwise in time, we also require that the pair $(\bp,\bfield)$ satisfies a condition called the \emph{spectral constraint}.
\begin{definition}[Spectral constraint]\label{def:sc}
	The pair $(\bp,\bfield)$ is said to satisfy the \emph{spectral constraint} if
	\begin{equation}\label{e:sc}
	\volav{\bp \abs{\nabla T}^2 + \bfield' wT } \geq 0 \qquad T \in \Tspace_n,
	\end{equation}
	where $w = -R \Delta^{-2} \Delta_h T$ solves~\eqref{e:w_and_T_eq} with the boundary conditions~\cref{e:w-bcs}.
\end{definition}

If the spectral constraint is satisfied, then it is possible to bound $\smash{\mean{wT}}$ from above in terms of $\bfield$, $\bp$, and another suitably constrained function $\lm:[0,1]\to \bR$.

\begin{proposition}[Bounding framework, IH1]\label{th:ih1-bounding-problem}
	Suppose that the pair $(\bp,\bfield)$ satisfies the spectral constraint and the boundary conditions in~\eqref{e:psi-bc-ih1}. Further, let $\lm \in L^2(0,1)$ be a nondecreasing function such that $\volav{\lm} = -1$.
	Then,
	\begin{equation*}
	\mean{wT} \leq 
	\tfrac12 + 
	\volav{ \tfrac{1}{4\bp} \abs{ \bfield' - \lm - \bp \left(z - \tfrac{1}{2}\right) }^2 - \bfield } =: U(\bfield,\lm,\bp)
	\end{equation*}
\end{proposition}
\begin{proof}
	A standard calculation using integration by parts, the incompressibility condition~\eqref{continuit}, and the boundary conditions on $\vu$ and $T$ yields
	\begin{align}
	\mean{wT} 
	&= \timeav{ \volav{wT} + \tfrac{\rm d}{ {\rm d} t}\mathcal{V}\{T(t)\}  }
	\nonumber\\[0.5ex]
	&= \tfrac12 +\mean{
		-\bp \abs{\nabla T}^2 - \bfield' wT 
		+ (\bfield' - \bp z) \partial_z T
		- \bfield
	}.
	\end{align}
	The infinite-time average on the right-hand side is bounded above by the largest value of the argument over the global attractor of the infinite-Prandtl-number Boussinesq equations~\eqref{e:governing-equations}. The minimum principle for the temperature field implies that the global attractor is contained in the set $\Tspace_+$ defined in~\eqref{e:T-positive-cone-ih1}. Consequently, we can estimate
	\begin{equation}\label{e:primal-form-ih1}
	\mean{wT}  \leq \tfrac12 + 
	\sup_{\substack{T \in \Tspace_+ \\ w = -R\Delta^{-2}\Delta_h T}}
	\volav{
		-\bp \abs{\nabla T}^2 - \bfield' wT 
		+ (\bfield' - \bp z) \partial_z T
		- \bfield
	}.
	\end{equation}
	This upper bound is finite if and only if unless the pair $(\bp,\bfield)$ satisfies the spectral constraint (cf. \cref{def:sc}), in which case the supremum over $T$ can be evaluated using a technical convex duality argument detailed in \cref{ss:duality-ih1}. The result of the argument is that the bound in~\cref{e:primal-form-ih1} is equivalent to
	\begin{equation}\label{e:dual-form-ih1}
	\mean{wT}  \leq \tfrac12 + 
	\inf_{ \substack{\lm \in L^2(0,1) \\ \lm \text{ \rm nondecreasing} \\ \langle \lm \rangle = -1}}
	\volav{ \tfrac{1}{4\bp} \abs{ \bfield' - \lm - \bp \left(z - \tfrac{1}{2}\right) }^2 - \bfield }.
	\end{equation}
	This inequality clearly implies the upper bound on $\smash{\mean{wT}}$ stated in the proposition, which is therefore proven.
\end{proof}

\begin{remark}
The function $\lm$ that arises when deriving~\cref{e:dual-form-ih1} from~\cref{e:primal-form-ih1} can be viewed as a Lagrange multiplier enforcing the pointwise nonnegativity of temperature fields in the set $\Tspace_+$. Further details regarding this interpretation (with slightly different notation) are given in~\cite[\S4.4]{arslan2021IH1}.
\end{remark}

\begin{remark} 
	The best upper bound on $\smash{\mean{wT}}$ provable with our approach is found upon minimizing the expression $U(\bfield,\lm,\bp)$ over all choices of $\bfield$, $\lm$ and $\bp$ that satisfy the conditions of \cref{th:ih1-bounding-problem}. This is hard to do analytically, but can be done computationally using a variety of numerical schemes (see~\cite{Fantuzzi2022} and references therein). We leave such computations to future work and focus on proving \cref{thm:ih1-main-result} by constructing suboptimal $\bfield$, $\lm$ and $\bp$ analytically.
\end{remark}

\subsection{Ans\"atze}\label{ss:ansatz-IH1}

To prove the upper bound on \wT, we seek $\bp>0$, $\bfield(z)$, and $\lm(z)$ that satisfy the conditions of \cref{th:ih1-bounding-problem} and make the quantity $U(\bp,\bfield,\lm)$ as small as possible. To simplify this task, we restrict $\bfield$ to take the form
\begin{equation}
\label{e:psi}
\bfield(z):= 
\begin{dcases}
1-\frac{z}{\delta} ,&  0\leq z \leq \delta\\
A\, \ln\!{\left(\frac{z (1-\delta)}{\delta (1-z)} \right)}, & \delta \leq z \leq 1-\varepsilon \\
A\, \ln\!\left(\frac{(1-\varepsilon)(1-\delta)}{\varepsilon\delta}\right) \, \left(\frac{1-z}{\varepsilon}\right), & 1-\varepsilon \leq z \leq 1
\end{dcases}
\end{equation}
and $\lm$ to be given by
\begin{equation}
\label{e:q_profile}
\lm(z):= 
\begin{dcases}
-\frac{1}{\delta} ,&  0\leq z \leq \delta,\\
0, & \delta \leq z \leq 1.
\end{dcases}
\end{equation}
These piecewise-defined functions, sketched in \cref{fig:psi_log_inf_pr}, are fully specified by the bottom boundary layer width $\delta \in (0,\frac12)$, the top boundary layer width $\varepsilon \in (0,\frac12)$, and the parameter $A > 0$ that determines the amplitude of $\bfield$ in the bulk of the layer.

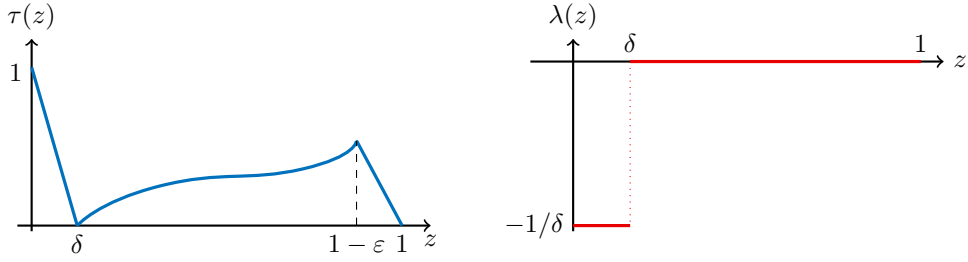
\begin{figure}
	\centering
	\input{./figs/tex/ih1-profiles-v2}
	\caption{Sketches of the functions $\bfield(z)$ in~\eqref{e:psi} and $\lm(z)$ in~\eqref{e:q_profile} used to prove \cref{thm:ih1-main-result}.}
	\label{fig:psi_log_inf_pr}
\end{figure}

We also fix 
\begin{equation}\label{e:b-choice-ih1}
\bp := {\volav{|\bfield'- \lm|^2}^{\frac12}}{\volav{|z - \tfrac12|^2}^{-\frac12}} = 2\sqrt{3} \volav{|\bfield'- \lm|^2}^{\frac12}.
\end{equation}
This choice is motivated by the desire to minimize the right-hand side of the inequality
\begin{equation}
\frac{1}{4\bp} \volav{ |  \bfield' - \lm - \bp \left(z - \tfrac{1}{2}\right)  |^2 } \leq \frac{\bp}{2} \volav{|z - \tfrac{1}{2}|^2 } + \frac{1}{2\bp} \volav{ |\bfield'- \lm|^2 },
\label{e:inequality_2norm_ih1}
\end{equation}
which is used later in \cref{th:ih1-U-bounds} by estimating from above the value of the bound $U(\bp,\bfield,\lm)$ for our choices of  $\bfield$ and $\lm$.

For any choice of the parameters $\delta$, $\varepsilon$, and $A$, the function $\bfield$ satisfies the boundary conditions in~\eqref{e:psi-bc-ih1}, while $\lm$ is nondecreasing and satisfies the normalization condition $\volav{\lm} = -1$. To establish~\cref{thm:ih1-main-result} using \cref{th:ih1-bounding-problem}, we only need to specify parameter values such that $U(\bp,\bfield,\lm) \leq \frac12 - O(\Ra^{-2})$ while ensuring that the pair $(\bp,\bfield)$ satisfies the spectral constraint. For the purposes of simplifying the algebra in what follows, we shall fix
\begin{equation}\label{e:A-choice}
A(\delta) = \frac{3\sqrt{15}}{20} \delta^{\frac32}
\end{equation}
from the outset. As explained in \cref{remark:ih1-A-value} below, this choice arises when insisting that the upper estimate on $U(\bp,\bfield,\lm)$ derived in \cref{th:ih1-U-bounds} be strictly less than $\frac12$ for some values of $\delta$ and $\varepsilon$, at least when all other constraints on these parameters are ignored.

\subsection{Preliminary estimates}\label{ss:lemmas-IH1}

We now derive a series of auxiliary results that make it simpler to specify the boundary layer widths $\delta$ and $\varepsilon$. The first result gives estimates on the value of $\bp$ in~\eqref{e:b-choice-ih1}.

\begin{lemma}[Estimates on $\bp$]\label{th:ih1-b-bounds}
	Let $\bfield(z)$, $\lm(z)$ and $\bp$ be given by~\eqref{e:psi}, \eqref{e:q_profile}, and~\eqref{e:b-choice-ih1} with $A$ specified by~\eqref{e:A-choice}. Suppose that the boundary layer widths $\delta$ and $\varepsilon$ satisfy
	\begin{subequations}
	    \begin{equation}
	    \delta \leq \frac16, \qquad
	    \varepsilon \leq \frac13, \qquad
		\delta \ln^{2}{\left(\frac{1}{\delta^2}\right)} \leq \varepsilon.
		\tag{\theequation a,b,c}
	    \end{equation}
	\label{e:eps_constraint_lemma}
	\end{subequations}
	Then,
	\begin{equation}
    \frac{9}{2\sqrt{5}} \delta   \leq \bp \leq \frac{9}{2} \delta.
	\label{e:b_upper_lower}
	\end{equation}
\end{lemma}

\begin{remark}
Condition~\subeqref{e:eps_constraint_lemma}{c} is key to prove the auxiliary results of this section. The other two restrictions, instead, are introduced to more easily keep track of constants in our estimates, which is necessary to obtain an explicit prefactor for the $O(\Ra^{-2})$ term. We have not attempted to optimize this prefactor.
\end{remark}

\begin{remark}\label{r:delta_less_eps}
Condition~\subeqref{e:eps_constraint_lemma}{c} implies that $\delta \leq \varepsilon$. We will use this fact often in the proofs of this section.
\end{remark}

\begin{proof}[Proof of \cref{th:ih1-b-bounds}]
    It suffices to estimate $\volav{ |\bfield' - \lm|^2 }^{\frac12}$ from above and below. For a lower bound, substitute our choices for $\bfield$ and $\lm$ to estimate
    \begin{equation*}
        \volav{ |\bfield' - \lm|^2} 
        = \int^{1}_{\delta}|\bfield'(z)|^2 \,\dz
        \geq \int^{1-\varepsilon}_{\delta} |\bfield'(z)|^2 \,\dz
        = A^2 \int^{1-\varepsilon}_{\delta} \left|\frac{1}{z} + \frac{1}{1-z}\right|^2 \,\dz.
    \end{equation*}
    Expanding the square, dropping the nonnegative term $\smash{\frac{2}{z(1-z)}}$, and recalling that $\varepsilon\leq \frac13$ and $\delta\leq \frac16 < \frac13$ by assumption, we can further estimate
    \begin{eqnarray}
        \volav{ |\bfield' - \lm|^2}
        \geq  A^2\left( 
                \int^{2/3}_{\delta} \frac{1}{z^2} \,\dz 
              + \int^{2/3}_{1/3} \frac{1}{(1-z)^2} \,\dz
              \right)
        = \frac{A^2}{\delta}.
        \label{e:lower_bound_L2_ih1}
    \end{eqnarray}
    Taking the square root of both sides and substituting for the value of $A$ from \eqref{e:A-choice} gives $\smash{\volav{|\bfield' - \lm|^2}^{1/2}} \geq  (3\sqrt{15}/20)\,\delta$, which combined with \eqref{e:b-choice-ih1} proves the lower bound on $\bp$ stated in~\eqref{e:b_upper_lower}.
    
    For the upper bound on $\bp$, recall that condition~\subeqref{e:eps_constraint_lemma}{c} implies $\delta \leq \varepsilon$, so $1-\varepsilon\leq 1-\delta\leq1$. Then,
    \begin{eqnarray*}
        \frac{\volav{ |\bfield'- \lm|^2 }}{A^2} 
        &=& \int^{1-\varepsilon}_{\delta} \left| \frac{1}{z} + \frac{1}{1-z} \right|^2 \textrm{d}z + \frac{1}{\varepsilon} \ln^{2}{\left(\frac{(1-\varepsilon)(1-\delta)}{\varepsilon\delta} \right)}
        \nonumber \\
        &\leq&  \int^{1-\delta}_{\delta} \left| \frac{1}{z} + \frac{1}{1-z} \right|^2 \textrm{d}z 
            + \frac{1}{\varepsilon} \ln^{2}\left(\frac{1}{\delta^2} \right).
    \end{eqnarray*}
    Using the inequality $(a+b)^2 \leq 2a^2 + 2b^2$ we can further estimate
    \begin{eqnarray*}
        \frac{\volav{ |\bfield'- \lm|^2 } }{A^2}
        &\leq& 
            2 \left(\int^{1-\delta}_{\delta} \frac{1}{z^2} \, \dz 
             + \int^{1-\delta}_{\delta} \frac{1}{(1-z)^2} \,\dz \right)  
             + \frac{1}{\varepsilon} \ln^{2}\left(\frac{1}{\delta^2} \right)
        \nonumber \\
        &=& \frac{4}{\delta}\left(\frac{1-2\delta}{1-\delta} \right) + \frac{1}{\varepsilon} \ln^{2}\left(\frac{1}{\delta^2} \right).
    \end{eqnarray*}
    Finally, we observe that $\frac{1-2\delta}{1-\delta} \leq 1$ for all $\delta\geq0$ and apply  \subeqref{e:eps_constraint_lemma}{c} to arrive at
    \begin{equation}\label{e:bf-lm-upper-bound}
        \volav{ |\bfield'- \lm|^2 }  \leq \frac{5A^2}{\delta}.
    \end{equation}
    Substituting our choice of $A$ from~\eqref{e:A-choice} and taking a square root gives $\smash{\smallvolav{|\bfield' - \lm|^2}^{1/2}} \leq  (3\sqrt{3}/4) \delta$, which combined with \eqref{e:b-choice-ih1} yields the upper bound on $\bp$ in~\eqref{e:b_upper_lower}.
\end{proof}

Our second auxiliary result estimates the upper bound $U(\bp,\bfield,\lm)$ on \wT\ from \cref{th:ih1-bounding-problem} in terms of the bottom boundary layer width $\delta$ alone.

\begin{lemma}[Estimates on $U(\bp,\bfield,\lm)$]\label{th:ih1-U-bounds}
	Let $\bfield(z)$, $\lm(z)$ and $\bp$ be given by~\eqref{e:psi}, \eqref{e:q_profile}, and~\eqref{e:b-choice-ih1} with $A$ specified by~\eqref{e:A-choice}. Suppose the boundary layer widths $\delta$ and $\varepsilon$ satisfy the conditions of \cref{th:ih1-b-bounds}. Then, 
	\begin{equation}
	U(\bp,\bfield,\lm) \leq \frac12 - \frac{\delta}{8}.
	\label{e:bound_eq}
	\end{equation}
\end{lemma}

\begin{proof}
    Using inequality \eqref{e:inequality_2norm_ih1}, our choice of $\bp$ from \eqref{e:b-choice-ih1}, and the upper bound on $\smallvolav{ |\bfield' - \lm|^2}$ from~\eqref{e:bf-lm-upper-bound} yields
    \begin{align*}
        U(\bp,\bfield,\lm) 
        \leq \frac12 +  \frac{\sqrt{3}}{6} \, \volav{ |\bfield' - \lm|^2}^{\frac12}   - \volav{\tau}
        \leq \frac12 +  \frac{\sqrt{15}A}{6} \,\delta^{-\frac12}   - \volav{\tau}.
    \end{align*}
    Moreover, since we have chosen $\bfield$ to be a non-negative function we can estimate $\volav{\tau} = \int_0^1 \tau(z) \,\dz \geq \int_0^\delta \tau(z) \,\dz = \frac{1}{2} \delta$
    %
    to obtain
    \begin{equation}\label{e:ih1-U-estimate-with-A}
        U(\bp,\bfield,\lm) \leq \frac12 +  \frac{\sqrt{15}A}{6}\,\delta^{-\frac12}   - \frac{1}{2} \delta.
    \end{equation}
    Substituting our choice of $A$ from~\eqref{e:A-choice} into this inequality gives~\eqref{e:bound_eq}.
\end{proof}

\begin{remark}\label{remark:ih1-A-value}
    The right-hand side of~\eqref{e:ih1-U-estimate-with-A} can be strictly smaller than $\frac12$ only if $A \lesssim \delta^{3/2}$. It is this observation that dictates the choice of $A$ in~\eqref{e:A-choice}. For any fixed value of $R$, one should choose $A \sim \delta^{3/2}$ with a (possibly $R$-dependent) prefactor that optimises the balance between the positive and negative terms, subject to constraints on $A$, $\delta$ and all other parameters that ensure the spectral constraint. To simplify our proof, however, we choose to fix this prefactor \textit{a priori} irrespective of $R$.
\end{remark}

Our final auxiliary result gives sufficient conditions on $\delta$ and $\varepsilon$ that ensure the spectral constraint (cf. \cref{def:sc}) is satisfied. 

\begin{lemma}[Sufficient conditions for the spectral constraint]\label{th:ih1-sc}
	Let $\bfield(z)$, $\lm(z)$ and $\bp$ be given by~\eqref{e:psi}, \eqref{e:q_profile}, and~\eqref{e:b-choice-ih1} with $A$ specified by~\eqref{e:A-choice}. Suppose the boundary layer widths $\delta$ and $\varepsilon$ satisfy the conditions of \cref{th:ih1-b-bounds}. Suppose further that
	\begin{subequations}
		\begin{equation}
		\delta \leq (24\sqrt{3})^2\, R^{-2} 
		\qquad\text{and}\qquad
		\varepsilon^3 \delta^\frac12 \, \ln^2\!\left(\frac{1}{\delta^{2}}\right) \leq 8\sqrt{3} R^{-1}.
		\tag{\theequation a,b}
		\end{equation}
		\label{e:delta_eps_cons}
	\end{subequations}
	Then, the pair $(\bp,\bfield)$ satisfies the spectral constraint.
\end{lemma}

The proof of this result relies on Hardy--Rellich inequalities established in Ref.~\cite{whitehead2011internal}, which extract a positive term from the \textit{a priori} indefinite term $\smallvolav{\bfield' w T}$. 
\begin{lemma}[Hardy--Rellich inequalities~\cite{whitehead2011internal}]\label{lemma:doering_estimate}
	Let $T,w:\Omega \to \bR$ be horizontally periodic functions such that $\Delta^2 w = -R \Delta_h T$ subject to velocity boundary conditions~\eqref{e:w-bcs}. Then,
	\begin{subequations}\label{doering_estimate}
		\begin{align}
	    \volav{\frac{w T}{z}} &\geq \frac{4}{R} \volav{\frac{w^2}{z^3}}
		&\text{and}&&
		\volav{\frac{w T}{1-z}}  &\geq \frac{4}{R} \volav{\frac{w^2}{(1-z)^3}}.
		\tag{\theequation a,b}
		\end{align}
	\end{subequations}
\end{lemma}
    
\begin{proof}[Proof of \cref{th:ih1-sc}]
    Let $\mathbbm{1}_{(a,b)}$ denote the indicator function of the interval $(a,b)$. Define the functions
    \begin{subequations}
    	\begin{gather}
    	\label{e:ih1-f}
    	f(z):= \left[ \frac{1}{\delta}+\frac{A}{z(1-z)} \right] \mathbbm{1}_{(0,\delta)}(z),\\
    	g(z):= \left[ \frac{1}{\varepsilon} \ln\!{\left(\frac{(1-\varepsilon)(1-\delta)}{\varepsilon\delta} \right)} +\frac{1}{z(1-z)} \right]\mathbbm{1}_{(1-\varepsilon,1)}(z).
    	\end{gather}
    \end{subequations}
    Given our choice of $\bfield$ from~\eqref{e:psi}, we can rewrite the spectral constraint as
    \begin{equation*}
    0 \leq\volav{ \bp|\nabla T|^2 + \bfield' w T  } = \mathcal{F}\{T\} + \mathcal{G}\{T\},
    \end{equation*}
    where
    \begin{subequations}
        \begin{gather}
        \mathcal{F}\{T\} :=  \frac{\bp}{2}  \volav{ |\nabla T|^2 } + A  \volav{ \frac{wT}{z} } -  \volav{ f(z) wT  },
        \\
        \mathcal{G}\{T\} := \frac{\bp}{2}  \volav{ |\nabla T|^2 }  + A  \volav{ \frac{wT}{1-z} } - A \volav{g(z) wT }.
        \end{gather}
    \end{subequations}
    Here, $w$ is determined as a function of $T$ by solving~\eqref{e:w_and_T_eq} subject to the boundary conditions in~\eqref{e:w-bcs}. We shall prove that $\mathcal{F}$ and $\mathcal{G}$ are individually non-negative. 
    
    First, let us consider $\mathcal{F}$. The Hardy--Rellich inequality~\subeqref{doering_estimate}{a} gives
    \begin{equation}\label{e:ih1-B-estimate}
        \mathcal{F}\{T\} \geq 
        \frac{\bp}{2}   \volav{ |\nabla T|^2}  + \frac{4 A}{R}  \volav{\frac{w^2}{z^3}}   - \volav{ f(z) wT } .
    \end{equation}
    Next, we estimate $\smallvolav{ f(z) wT }$. 
    Since $T$ vanishes at $z=0$ by virtue of the thermal boundary conditions in~\eqref{bc_T_IH1}, we can use the fundamental theorem of calculus and the Cauchy-Schwarz inequality to estimate $|T(\cdot,z)|\leq\smash{\sqrt{z}(\int^{1}_{0}|\nabla T|^2 \textrm{d}z )^{1/2}}$. Squaring both sides and taking the horizontal average of which gives $\volav{T^2}_h \leq z \volav{|\nabla T|^2}$.
    Then, use of the Cauchy--Schwarz inequality, substitution for $\volav{T^2}_h$ and Youngs inequality gives
    \begin{align*}
         \volav{ f(z) wT } \leq \int^{1}_{0} |f(z)| \volav{|wT|}_h \textrm{d}z &\leq
         \int^{1}_0 |f(z)| \volav{w^2}^{1/2}_h \volav{T^2}^{1/2}_h \textrm{d}z \nonumber \\
         &\leq
         \volav{|\nabla T|^2}^{\frac12}\int^{1}_{0} f(z) z^2 \volav{\frac{w^2}{z^{3}}}^{\frac12}_h \textrm{d}z  \nonumber \\     
         &\leq
         \volav{|\nabla T|^2}^{\frac12}  \volav{ f(z)^2 z^{4} }^{\frac12}   \volav{\frac{w^2}{z^3} }^{\frac12}\nonumber\\
         &\leq
         \frac{\bp}{2} \volav{|\nabla T|^2 } + \frac{1}{2\bp}  \volav{ f(z)^2 z^{4}} \volav{\frac{w^2}{z^3} } .
    \end{align*}
    Upon using the lower bound on $\bp$ from \eqref{e:b_upper_lower} to estimate the last term from above we obtain
    \begin{equation}
        \volav{ f(z) wT } \leq \frac{\bp}{2} \volav{|\nabla T|^2 } + \frac{\sqrt{5}}{9\delta} \volav{ f(z)^2 z^{4}} \volav{\frac{w^2}{z^3} } .
    \end{equation}
    This can be substituted into~\eqref{e:ih1-B-estimate} along with our choice of $A$ from~\eqref{e:A-choice} to find
    \begin{equation}
        \mathcal{F}\{T\} \geq \frac{\sqrt{5}}{9}\delta^{\frac32}\left( 
        \frac{27\sqrt{3}}{5 R} \, 
        - \frac{\volav{ f(z)^2 z^{4}}}{\delta^{\frac52}}\right) 
        \volav{ \frac{w^2}{z^3} }.
    \end{equation}
    To conclude, we show that the term in parentheses is nonnegative when $\delta$ satisfies $\delta \leq \frac16$ and~\subeqref{e:delta_eps_cons}{a}.
    To do this, we observe that the function $f(z)$ in~\eqref{e:ih1-f} is nonnegative function and that it is nonzero only if $z \leq \delta$. We can therefore bound it from above on the interval $(0,\delta)$ using the estimates $\frac{1}{1-z} \leq  \frac{1}{1-\delta} \leq \frac65$ and, consequently, obtain
     \begin{equation*}
        \frac{\volav{ f(z)^2 z^{4}}}{\delta^{\frac52}}
        \leq \frac{1}{\delta^{\frac52}}\int^{\delta}_{0} \left(\frac{1}{\delta} + \frac{6A}{5z} \right)^2 z^4 \, \dz
        = \bigg(\frac15 + \frac{9\sqrt{15}}{100} \delta^{\frac32} + \frac{81}{500} \delta^3\bigg)  \delta^\frac12.
    \end{equation*}
    Using the assumption that $\delta \leq \frac16$ to estimate the expression in parentheses by its value at $\delta=\frac16$, followed by an application of assumption~\subeqref{e:delta_eps_cons}{a} to estimate the remaining $\delta^{1/2}$ term in terms of $\Ra$ we arrive at the desired inequality
    \begin{equation*}
        \frac{\volav{ f(z)^2 z^{4}}}{\delta^{\frac52}}
        \leq \frac{9\delta^\frac12}{40}
        \leq \frac{27\sqrt{3}}{5 \Ra}.
    \end{equation*}

    Analogous arguments show that $\mathcal{G}\{T\}$ is nonnegative. Using the Hardy--Rellich inequality~\subeqref{doering_estimate}{b} we have
    \begin{equation}\label{e:ih1-T-estimate}
    \mathcal{G}\{T\} \geq
    \frac{\bp}{2} \volav{  |\nabla T|^2 } + \frac{4 A}{R} \volav{ \frac{|w|^2}{(1-z)^3} } - A \volav{ g(z) wT } .
    \end{equation}
    To estimate the last term, we use (in order) the inequality $\volav{T^2}_h\leq \smash{(1-z)\volav{ |\nabla T|^2} }$, the Cauchy--Schwarz inequality, and Young's inequality:
    \begin{align*}
        \volav{g(z) w T} \leq \int^{1}_{0} |g(z)| \volav{|wT|}_h \textrm{d}z &\leq \int^{1}_{0} |g(z)| \volav{w^2}^{1/2}_h \volav{T^2}^{1/2}_h \textrm{d}z \nonumber \\
         &\leq  \volav{|\nabla T |^2}^{\frac12} \int^{1}_{0} g(z)(1-z)^2 \volav{\frac{w^2}{(1-z)^{3}} }_{h}^{\frac12} \textrm{d}z
        \nonumber \\
        &\leq \volav{|\nabla T|^2}^{\frac12} \volav{ g(z)^2(1-z)^4 }^{\frac12} \volav{ \frac{w^2}{(1-z)^3}}^{\frac12} 
        \nonumber \\
        &\leq \frac{\bp}{2A} \volav{ |\nabla T|^2}  + \frac{A}{2\bp} \volav{ g(z)^2 (1-z)^4 } \volav{ \frac{w^2}{(1-z)^3} }   .    
    \end{align*}
    Using the lower bound on $\bp$ from \eqref{e:b_upper_lower} gives
    \begin{equation}
        \volav{g(z) w T} \leq
        \frac{\bp}{2A} \volav{ |\nabla T|^2}  + \frac{\sqrt{5}A}{9 \delta} \volav{ g(z)^2 (1-z)^4 } \volav{ \frac{w^2}{(1-z)^3} },
    \end{equation}
    which can be substituted into \eqref{e:ih1-T-estimate} along with the value of $A$ from~\eqref{e:A-choice} to obtain
    \begin{equation}
        \mathcal{G}\{T\} \geq 
        \frac{3 \delta^\frac32}{16 \sqrt{5}}  \left( \frac{16\sqrt{3}}{\Ra}  -  \delta^{\frac12} \volav{ g(z)^2 (1-z)^4 } \right)\volav{  \frac{w^2}{(1-z)^3} }.
        \label{e:T_1}
    \end{equation}
    To conclude the argument we show that the term in parentheses is non-negative. To demonstrate this, we first estimate $g(z)$ on the interval $(1-\varepsilon,1)$ from above using the assumption that $\varepsilon\leq \frac13$, so $\frac23\leq z \leq1$ and $\frac{1}{z(1-z)} \leq \frac{3}{2(1-z)}$. Thus,
    \begin{equation*}
        g(z) \leq \frac{1}{\varepsilon} \ln{\left(\frac{1}{\delta^{2}}\right)} + \frac{3}{2(1-z)} \qquad \forall z \in (1-\varepsilon,1).
    \end{equation*}
    Then, we use the assumptions $\delta\leq\frac16$ and $\delta\leq \varepsilon$ (cf. \cref{r:delta_less_eps}) to observe that $\ln{(\frac{1}{\varepsilon\delta})} \leq \ln{(\frac{1}{\delta^2})}\leq  \ln^{2}{(\frac{1}{\delta^2})} $ and $\frac34 \leq  \frac34 \ln^{2}{(\frac{1}{\delta^2})}$. Combining these estimates with the upper bound on $g$ derived above gives
    \begin{align*}
       \delta^{\frac12} \volav{ g(z)^2 (1-z)^4 }
       &\leq 
       \delta^{\frac12}\int^{1}_{1-\varepsilon} \left( \frac{1}{\varepsilon} \ln{\left(\frac{1}{\delta^{2}}\right)} + \frac{3}{2(1-z)} \right)^2 (1-z)^4 \, \textrm{d}z \nonumber \\
       & = \delta^{\frac12}\varepsilon^3 \left(  \frac15 \ln^{2}{\left(\frac{1}{\delta^2}\right) } +  \frac34\ln{\left(\frac{1}{\delta^{2}}\right)} + \frac34  \right) \nonumber\\
      & \leq   2\delta^{\frac12} \varepsilon^3 \ln^{2}{\left( \frac{1}{\delta^{2}} \right)} \nonumber \\
      (\text{\textrm{by} \subeqref{e:delta_eps_cons}{b}}) \qquad & \leq \frac{16\sqrt{3}}{R},
    \end{align*}
    as desired. This concludes the proof of \cref{th:ih1-sc}.
\end{proof}

\subsection{Proof of \texorpdfstring{ \cref{thm:ih1-main-result}}{Theorem 1} }
\label{ss:proof-IH1}

It is now straightforward to prove the upper bound on \wT\ by specifying boundary layer widths $\delta$ and $\varepsilon$ that satisfy the conditions of \cref{th:ih1-b-bounds,,th:ih1-U-bounds,,th:ih1-sc}. 

\begin{figure}
    \centering
    \includegraphics[scale=1]{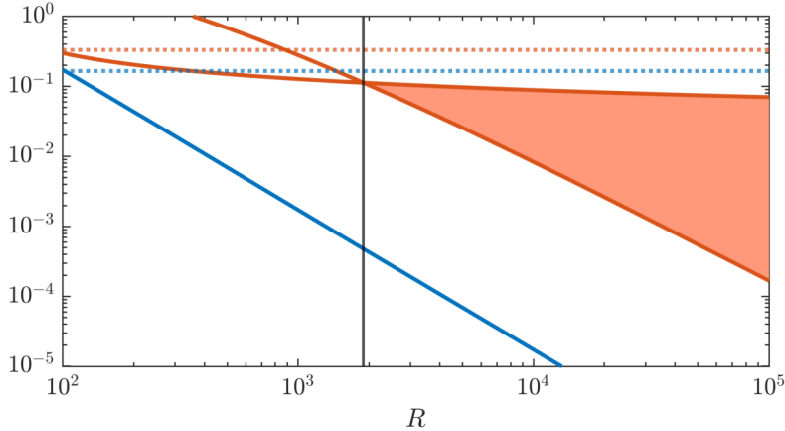}
    \caption{Variation with \Ra\ of the allowed values for the bottom boundary layer width $\varepsilon$ (shaded region), determined by condition~\subeqref{e:eps_constraint_lemma}{c} in \cref{th:ih1-b-bounds} and condition~\subeqref{e:delta_eps_cons}{b} in~\cref{th:ih1-sc} when the bottom boundary layer width $\delta$ ({\color{matlabblue}\solidrule}) is chosen as in~\eqref{e:choice_vars_final}. Also shown are the uniform upper bounds $\delta\leq \frac16$ ({\color{matlabblue}\dottedrule}) and $\varepsilon\leq\frac13$ ({\color{matlabred}\dottedrule}) imposed on these variables. A black vertical line marks the Rayleigh number $\Ra_0 \simeq 1891.35$ above which all constraints on $\delta$ and $\varepsilon$ are satisfied.  }
    \label{fig:ih1_conditions}
\end{figure}

Since the estimate for the resulting upper bound obtained in \cref{th:ih1-U-bounds} is minimized when $\delta$ is as large as possible, we choose the largest value consistent with \subeqref{e:delta_eps_cons}{a},
\begin{equation}
\label{e:choice_vars_final}
\delta = (24\sqrt{3})^2\, R^{-2}.
\end{equation}
With this choice of $\delta$, conditions~\subeqref{e:eps_constraint_lemma}{c} and~\subeqref{e:delta_eps_cons}{b} require $\varepsilon$ to satisfy %
\begin{equation}
    \frac{27\,648}{R^{2}}\,\ln^2{\left(\frac{R}{24\sqrt{3}} 
    \right)} \leq \varepsilon \leq  \left(\frac{1}{48}\right)^{\frac13} \ln^{-\frac23}{\left(\frac{R}{24\sqrt{3}} 
    \right)},
    \label{e:eps_range}
\end{equation}
which is possible for $\Ra>\Ra_0 \simeq 1891.35$ (cf. \cref{fig:ih1_conditions}). For $\Ra > \Ra_0$, any choice of $\varepsilon$ in this range is feasible. The optimal value could be determined at the expense of more complicated algebra either by optimizing the full bound $U(\bp,\bfield,\lm)$, or by deriving better $\varepsilon$-dependent estimates for it. However, we expect that any $\varepsilon$-dependent terms will contribute only higher-order corrections to our bound on $\mean{wT}$. 

To conclude the proof of \cref{thm:ih1-main-result}, there remains to verify that our choice of $\delta$ is no larger than
$\frac16$ and that any $\varepsilon$ satisfying~\eqref{e:eps_range} is no larger than $\frac13$. It is easily checked that both conditions hold when $\Ra \geq \Ra_0$ (see \cref{fig:ih1_conditions} for an illustration). For all such values of $R$, therefore, \cref{th:ih1-U-bounds} and our choice of $\delta$ yield the upper bound $\mean{wT} \leq U(\bp,\bfield,\lm) \leq \frac12 - cR^{-2}$ with $c = 216$.

\section{Bounds for the IH3 configuration}\label{sec:IH3}

We now move on to studying IH convection in the IH3 configuration, where the top boundary is maintained at constant (zero) temperature and the bottom boundary is insulating. 
First, in \cref{ss:bm-IH3}, we derive a bounding framework for $\smash{\mean{wT}}$ following steps similar to those used for the IH1 case (cf. \cref{ss:bm-IH1}). In \cref{ss:ansatze-IH3} we present ans\"atze for $\bfield$ and $\lm$, with which we obtain crucial estimates in \cref{ss:lemmas_ih3} , which give the bound in \cref{ss:analtic_bound_ih3}.  Throughout this section, $T$ belongs to  $\Tspace_3$. Observe that this changes the set of temperature fields over which the spectral constraint in \cref{def:sc} is imposed.
The notation $\Tspace_+$ still denotes the subset of temperature fields in $\Tspace_3$ that are nonnegative pointwise almost everywhere.

\subsection{Bounding framework}\label{ss:bm-IH3}
Upper bounds on $\smash{\mean{wT}}$ for the IH3 configuration can be derived using a quadratic auxiliary function $\mathcal{V}\{T\}$ similar to that used for the IH1 case. Precisely, we still take $\mathcal{V}$ to be defined as in~\eqref{e:V-IH1}, where the positive constant $\bp$ and the piecewise-differentiable square-integrable function $\bfield(z)$ are tunable parameters. However, this time we impose only the boundary condition
\begin{equation}\label{e:psi-bc-ih3}
    \bfield(1)=0.
\end{equation}
These changes result in the following family of parametrized upper bounds on $\mean{wT}$.

\begin{proposition}[Bounding framework, IH3]\label{th:ih3-bounding-problem}
	Suppose that the pair $(\bp,\bfield)$ satisfies the spectral constraint (cf. \cref{def:sc}) and the boundary condition in~\eqref{e:psi-bc-ih3}. Further, let $\lm \in L^2(0,1)$ be a nondecreasing function such that $\lm(0) = -1$.
	Then,
	\begin{equation*}
	\mean{wT} \leq 
	\tfrac12 + 
	\volav{ \tfrac{1}{4\bp} \abs{ \bfield' - \lm - \bp z }^2 - \bfield } =: U(\bfield,\lm,\bp)
	\end{equation*}
\end{proposition}

\begin{proof}
    Proceeding as in the proof of \cref{th:ih1-bounding-problem} shows that
    \begin{equation}\label{e:primal-form-ih3}
        \mean{wT} \leq \tfrac12  
        + \sup_{\substack{T \in \Tspace_+ \\ w = -R\Delta^{-2}\Delta_h T}}\volav{  
            - b \abs{\nabla T}^2 - \bfield' w T  + \left(\bfield - \bp z + 1 \right)\partial_z T 
            - \bfield
            }.
    \end{equation}
    The supremum on the right-hand side can be evaluated using the convex duality argument summarized in \cref{ss:duality-ih3}, leading to the equivalent inequality
    \begin{equation}\label{e:dual-form-ih3}
        \mean{wT} \leq \tfrac12  
        + \inf_{ \substack{\lm \in L^2(0,1) \\ \lm \text{ \rm nondecreasing} \\ \lm \geq -1}}
        \volav{ \tfrac{1}{4\bp} \abs{ \bfield' - \lm - \bp z }^2 - \bfield }.
    \end{equation}
    This clearly implies the upper bound stated in the proposition.
\end{proof}

\subsection{Ans\"atze}\label{ss:ansatze-IH3}
The procedure for the proof of an upper bound on $\mean{wT}$ is the same as that employed for isothermal boundaries. We construct $\bp >0$, $\bfield(z)$ and $\lm(z)$ that satisfy the conditions of \cref{th:ih3-bounding-problem}, while trying to minimize the corresponding bound $U(\bp,\bfield,\lm)$. Due to the Neumann boundary condition on $T$ at $z=0$, we can no longer employ the Poincar\'e estimates used in \cref{ss:lemmas-IH1} to control the sign-indefinite term in the spectral constraint at the bottom boundary. Instead we modify $\bfield(z)$ in $(\delta,\frac12)$ to increase slower than logarithmically in $z$ and use results established in~\cite{Whitehead2014mixed}.
The function $\bfield(z)$ is hence chosen to have the form
\begin{equation}
\label{e:psi_3}
    \bfield(z):= 
\begin{dcases}
    \delta -z ,&  0\leq z \leq \delta,\\
    \frac{A}{1-\alpha}\left(z^{1-\alpha} - \delta^{1-\alpha} \right) - A \ln{\left(\frac{1-z}{1-\delta}\right)},  & \delta \leq z \leq 1-\varepsilon, \\
   \frac{A\,B}{\varepsilon} \left(1-z\right), & 1-\varepsilon \leq z \leq 1,
\end{dcases}
\end{equation}
where
\begin{align}
    B=B(\epsilon,\delta,\alpha) := \frac{(1-\varepsilon)^{1-\alpha}-\delta^{1-\alpha}}{1-\alpha} + \ln{\left(\frac{1-\delta}{\varepsilon} \right)}.
    \label{e:B_form}
\end{align}
On the other hand, the Lagrange multiplier $\lm(z)$ is still chosen to be
\begin{equation}
\label{e:q_3}
    \lm(z):= 
\begin{dcases}
    -1, &  0\leq z \leq \delta,\\
    0  , & \delta \leq z \leq 1 .
\end{dcases}
\end{equation}
These piecewise functions, sketched in \cref{fig:psi_log_inf_pr_ih3}, are fully specified by the bottom and top boundary layer widths  $\delta, \varepsilon \in (0,\frac12)$, the constant $A>0$, and the exponent $\alpha \in (0,1)$ driving the behaviour of $\bfield(z)$ in the bulk. 

For $\bp$ we take
\begin{equation}
    \bp := \volav{ |\bfield' - \lm|^2}^{\frac12} \volav{z^2}^{-\frac12}  = \sqrt{3} \volav{ |\bfield' - \lm|^2 }^{\frac12}.
    \label{e:b_choice_ih3}
\end{equation}
This choice is motivated by minimizing the right hand side of the estimate
\begin{equation}
    \frac{1}{4\bp} \volav{ |  \bfield' - \lm  - \bp z|^2} \leq \frac{\bp}{2} \volav{ z^2}  + \frac{1}{2\bp} \volav{ |\bfield'- \lm |^2 } = \frac{1}{\sqrt{3}} \volav{ |\bfield'- \lm |^2 }^{\frac12} = \frac{\beta}{3},
    \label{e:est_2norm_starting_point}
\end{equation}
which is used in \cref{lemma:U_ih3} below to estimate the value of the bound $U(\bp,\bfield,\lm)$ from above when $\bfield$ and $\lm$ are define by \eqref{e:psi_3} and \eqref{e:q_3} respectively.

For any choice of the parameters $\delta$, $\varepsilon$, $A$, and $\alpha$, the function $\bfield$ satisfies the boundary conditions in \eqref{e:psi-bc-ih3}, while $\lm$ is nondecreasing and satisfies the condition $\lm(0)=-1$. Thus, to establish \cref{thm:ih3} using \cref{th:ih3-bounding-problem} we need only specify parameter values such that $U(\bp,\bfield,\lm) \leq \frac12 - O(R^{-4})$ while ensuring that $(\bp,\bfield)$ satisfy the spectral constraint. For the purposes of simplifying the algebra in what follows, we shall fix
\begin{equation}
    A(\delta,\alpha) = \frac{2\sqrt{3}}{9} \sqrt{2\alpha-1}\, \delta^{\alpha+\frac32}
    \label{e:A_choice_ih3} 
\end{equation}
from the outset. This choice arises when insisting that the upper estimate on $U(\bp,\bfield,\lm)$ derived in \cref{lemma:U_ih3} below should be strictly less than $\frac12$ for suitable choices of $\delta$ and $\varepsilon$, at least when all other constraints on these parameters are ignored.
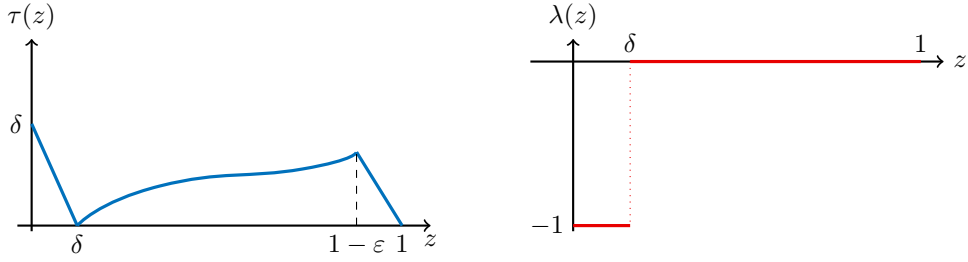
\begin{figure}
    \centering
    \input{figs/tex/ih3-profiles-v2}
    \caption{Sketch of the functions $\bfield(z)$ \eqref{e:psi_3} and $\lm(z)$ in \eqref{e:q_3} used to prove \cref{thm:ih3}.  }
    \label{fig:psi_log_inf_pr_ih3}
\end{figure}
%

\subsection{Preliminary Estimates}\label{ss:lemmas_ih3}
We now derive auxiliary results that simplify the choice of the exponent $\alpha$ and of the boundary widths $\delta$ and $\varepsilon$. The first gives estimates on the value of $\bp$ in~\eqref{e:b_choice_ih3}.

\begin{lemma}[Estimates on $\bp$]\label{lemma:b_ih3} Let $\bfield(z)$, $\lm(z)$ and $\bp$ be given by \eqref{e:psi_3}, \eqref{e:q_3} and \eqref{e:b_choice_ih3} with $A$ specified in \eqref{e:A_choice_ih3}. Suppose  that $\alpha \in (\frac12,1)$ and the boundary layer widths $\delta$ and $\varepsilon$ satisfy
	\begin{subequations}
	    \begin{equation}
	    \delta \leq \frac13 \left(\frac12\right)^{\frac{1}{2\alpha -1}}, \qquad
	    \varepsilon \leq \frac13, \qquad
	    (2\alpha-1)\,\delta^{2\alpha-1}B(\epsilon,\delta,\alpha)^2 \leq \varepsilon.
		\tag{\theequation a,b,c}
	    \end{equation}
	   \label{e:eps_bound_conditions}
	\end{subequations}
Then
\begin{equation}
    \frac{\sqrt{2}}{3}\, \delta^2   \leq \bp \leq \frac43\, \delta^2.
    \label{e:b_lower_upper_ih3}
\end{equation}
\end{lemma}

\begin{remark}
Condition \subeqref{e:eps_bound_conditions}{a} and the bounds on $\alpha$ imposed in the Lemma imply that $0 \leq \delta \leq \frac16$. These uniform bounds will be used repeatedly in the following proofs.
\end{remark}

\begin{proof}[Proof of \cref{lemma:b_ih3}]
It suffices to estimate $\volav{|\bfield'-\lm|^2}^{\frac12}$ from above and below. For a lower bound, we can substitute our choices of $\bfield$ and $\lm$ and then estimate  
\begin{eqnarray}
    \volav{ |\bfield' - \lm |^2} = \int^{1}_{\delta} |\bfield'(z)|^2\textrm{d}z \geq \int^{1-\varepsilon}_{\delta} |\bfield'(z)|^2\textrm{d}z
    = A^2 \int^{1-\varepsilon}_{\delta} \left|\frac{1}{z^{\alpha}} + \frac{1}{1-z} \right|^2 \textrm{d}z . \nonumber 
\end{eqnarray}
Dropping the positive term $1/(1-z)$ from the integrand and integrating the rest gives
\begin{equation*}
     \volav{ |\bfield' - \lm |^2} \geq A^2 \left(\frac{\delta^{1-2\alpha}-(1-\varepsilon)^{1-2\alpha}}{2\alpha-1}  \right).
\end{equation*}
For every $\alpha \in (\frac12,1)$, the second term inside the parentheses can be estimated upon observing that constraints \subeqref{e:eps_bound_conditions}{a-b} imply
\begin{equation}
    \label{e:est_1_eps}
    (1-\varepsilon)^{1-2\alpha}\leq \left(\frac23 \right)^{1-2\alpha }  \leq   \frac{1}{2^{2\alpha}}  \delta^{1-2\alpha} \leq \frac12  \delta^{1-2\alpha} .
\end{equation}
Thus, we obtain
\begin{equation} \label{e:lower_bound_L2_ih3}
    \volav{|\bfield' - \lm|^2} 
    \geq A^2 \frac{\frac12 \delta^{1-2\alpha}}{2\alpha - 1} .
\end{equation}
Taking the square root of \eqref{e:lower_bound_L2_ih3} and using \eqref{e:A_choice_ih3} gives $\volav{| \bfield' - \lm|^2}^{\frac12} \geq \frac{\sqrt{6}}{9}\, \delta^2$, which combined with \eqref{e:b_choice_ih3} proves the lower bound on $\bp$ stated in \eqref{e:b_lower_upper_ih3}.

To prove the upper bound on $\bp$, we start by using the inequality $(a+b)^2 \leq 2a^2 + 2b^2$, evaluating exactly the integral of $\bfield'-\lm$, and  dropping the negative terms to get
\begin{align}
    \frac{\volav{ |\bfield' - \lm|^2}}{A^2}  &= \int^{1-\varepsilon}_{\delta} \left|\frac{1}{z^{\alpha}} + \frac{1}{1-z} \right|^2 \textrm{d}z + \frac{ B^2}{\varepsilon} \nonumber \\
    &\leq \int^{1-\varepsilon}_{\delta} \frac{2}{z^{2\alpha}} + \frac{2}{(1-z)^2} \textrm{d}z +  \frac{ B^2}{\varepsilon} \nonumber \\ 
 \text{(since $\alpha > \frac12$)} \quad   &\leq \frac{2\delta^{1-2\alpha}}{2\alpha-1} + \frac{2}{\varepsilon} + \frac{B^2}{\varepsilon}. 
    \label{e:upper_1_est_ih3}
\end{align}
Using assumption \subeqref{e:eps_bound_conditions}{c}, the second and final term in \eqref{e:upper_1_est_ih3} can be estimated from above to arrive at
\begin{equation}
      \frac{\volav{|\bfield' - \lm|^2}}{A^2} \leq   \frac{\delta^{1-2\alpha}}{2\alpha -1}\left(2+ \frac{2}{B^2}+1 \right).
      \label{e:2_norm_est_ih3}
\end{equation}
Next, we observe that for all $\varepsilon \in (0,\frac13)$, $\delta \in (0,\frac16)$, and $\alpha \in (\frac12,1)$ we can estimate
\begin{equation*}
    B(\varepsilon,\delta,\alpha) \geq B\left( \frac13, \frac16, \frac12 \right) = \frac{\sqrt{6}}{3} + \ln{\left(\frac52\right)} > \sqrt{2},
\end{equation*}
so $B^{-2} \leq 1/2$. Using this estimate in \eqref{e:2_norm_est_ih3}, taking a square root, and substituting in the value of $A$ given in \eqref{e:A_choice_ih3} leads to the inequality $\volav{ |\bfield'- \lm|^2}^{1/2} \leq (4\sqrt{3}/9)\, \delta^2$. Combining this with \eqref{e:b_choice_ih3} yields the upper bound on $\bp$ stated in \eqref{e:b_lower_upper_ih3} and concludes the proof of \cref{lemma:b_ih3}.
\end{proof}

The second auxiliary result of this section estimates the upper bound $U(\bp,\bfield,\lm)$ on $\mean{wT}$ given by \cref{th:ih3-bounding-problem} using only the bottom boundary layer width $\delta$.

\begin{lemma}[Estimates on $U(\bp,\bfield,\lm)$]\label{lemma:U_ih3} Let $\bfield(z)$, $\lm(z)$ and $\bp$ be specified by \eqref{e:psi_3}, \eqref{e:q_3} and \eqref{e:b_choice_ih3} with $A$ given by \eqref{e:A_choice_ih3}. Suppose $\alpha$ and the boundary layer widths $\delta$, $\varepsilon$ satisfy the conditions of \cref{lemma:b_ih3}. Then,
  \begin{equation}
        U(\bp,\bfield,\lm) \leq \frac12 - \frac{\delta^2}{18}.
        \label{e:U_final_ih3}
    \end{equation}
\end{lemma}

\begin{proof}
Inequality \eqref{e:est_2norm_starting_point} and the upper bound on $\beta$ from \eqref{e:b_lower_upper_ih3} give
\begin{equation}\label{e:U_lemma_v3}
    U(\bp,\bfield,\lm) \leq \frac12 + \frac{\beta}{3} - \volav{\bfield} \leq \frac12 + \frac{4}{9}\delta^2 - \volav{\bfield}. 
\end{equation}
The result follows upon observing that $\bfield$ is non-negative, so $\volav{\bfield} = \int^{1}_0 \bfield(z) \textrm{d}z \geq \int^{\delta}_0 \bfield(z) \textrm{d}z = \frac12 \delta^2$. 
\end{proof}

\begin{remark}
The right hand side of \eqref{e:U_lemma_v3} can be strictly smaller than $\frac12$ when $\delta^2$ is small only if $A \lesssim \delta^{\alpha + 3/2}$. This observation dictates the choice of $A$ in \eqref{e:A_choice_ih3}. For any fixed value of $R$, one should choose $A \sim \delta^{\alpha + 3/2}$ with a (possibly $R$-dependent) prefactor that optimises the balance between the positive and negative terms, subject to constraints on $A$, $\delta$ and all other parameters that ensure the spectral constraint. To simplify our proof, however, we choose to fix this prefactor \textit{a priori} irrespective of $R$.

\end{remark}

Our final auxiliary result gives the sufficient conditions on $\delta$ and $\varepsilon$ that ensure the spectral constraint (cf. \cref{def:sc}) is satisfied. 

\begin{lemma}[Sufficient conditions for spectral constraint] \label{lemma:spectral_ih3} Let $\bfield(z)$, $\lm(z)$ and $\bp$ be specified by \eqref{e:psi_3}, \eqref{e:q_3} and \eqref{e:b_choice_ih3} with $A$ given by \eqref{e:A_choice_ih3}. Suppose that $\alpha$ and the boundary layer widths satisfy the conditions of \cref{lemma:b_ih3}. Further, let
\begin{equation}
     h(\alpha) =2(2\alpha -1 )(1-\alpha^2).
    \label{e:h_a_lem}
\end{equation}
and suppose that
\begin{subequations}
    \begin{gather}
        \delta \leq \varepsilon,\\
        \delta \leq   h(\alpha)  \,   R^{-2},\\
        \frac{\sqrt{2\alpha - 1}}{(1-\alpha)^2}\,\delta^{\alpha-\frac12}\varepsilon^3\ln^2{(\delta^{-1})} \leq \frac43 \sqrt{6} R^{-1}.
    \end{gather}
    \label{e:delta_eps_lemma}
\end{subequations}
Then, the pair $(\bp,\bfield)$ satisfies the spectral constraint.
\end{lemma}

Unlike the analogous result obtained in \cref{ss:lemmas-IH1}, \cref{lemma:spectral_ih3} cannot be proven using only the Hardy--Rellich inequalities stated in \cref{lemma:doering_estimate}. The lack of a fixed boundary temperature at $z=0$, makes it impossible to gain sufficient control on the contribution of the bottom boundary layer to the quadratic form in~\cref{e:sc}. This difficulty can be overcome using
the following result, obtained as a particular case of a more general analysis by Whitehead and Wittenberg~\citep[][Eqs. (59) \& (77)]{Whitehead2014mixed}, which upon setting (in their notation) $\nu_1 = \frac{\alpha}{2} - \nu_2$ and $\nu_2 = \frac12\left(-1+\sqrt{2(1-\alpha^2)}\right)$.
\begin{lemma}[Adapted from~\cite{Whitehead2014mixed}]\label{lemma_3_main}
Fix $\alpha \in (\frac12,1)$ and $\mu,R>0$. Suppose $\varphi(z): [0,1] \rightarrow \bR$, be a non-negative function, satisfying
\begin{equation}\label{e:lemma_cons_WW}
    \volav{ \varphi(z) z^{1+\alpha/2}}^2  \leq \frac{\mu}{R} \frac{\sqrt{2}(3-\alpha)(2+\alpha) \sqrt{1-\alpha^2}}{3+\alpha}.
\end{equation}
Then, for every $w$ and $T$ solving~\eqref{e:w_and_T_eq} subject to the boundary conditions \eqref{e:w-bcs},
\begin{equation}
     \volav{\frac12 |\nabla T|^2  +
      \mu\frac{wT }{z^{\alpha}}  - \varphi(z) wT }\geq 0.
\end{equation}
\end{lemma}

We are now ready to prove \cref{lemma:spectral_ih3}.

\begin{proof}[Proof of \cref{lemma:spectral_ih3}]
Let $\mathbbm{1}_{(a,b)}$ denote the indicator function of the interval $(a,b)$ and define the functions
\begin{subequations}
\begin{gather}
f(z) :=  \left[1 + \frac{A(\delta,\alpha)}{z^{\alpha}} + \frac{A(\delta,\alpha)}{1-z} \right] \mathbbm{1}_{(0,\delta)}(z) ,\label{e:j}\\
g(z) :=  \left[\frac{B(\varepsilon,\delta,\alpha)}{\varepsilon}+ \frac{1}{z^\alpha} + \frac{1}{1-z} \right] \mathbbm{1}_{(1-\varepsilon,1)}(z)  . \label{e:k}
\end{gather}
\end{subequations}
Given our choice of $\bfield$, we can rewrite the spectral constraint as 
\begin{equation*}
    0 \leq \volav{\bp |\nabla T|^2 + \bfield' w T } = \mathcal{F}\{T\} + \mathcal{G}\{T\}, 
\end{equation*}
where 
\begin{subequations}
\begin{gather}
    \mathcal{F}\{T\} := \frac{\bp}{2} \volav{|\nabla T|^2}    + A \volav{\frac{w T }{z^{\alpha}}}\,  -  \volav{f(z)\, w T }, \\
    \mathcal{G}\{T\} := \frac{\bp}{2} \volav{|\nabla T|^2}    + A \volav{ \frac{w T }{1-z}} -  \,   A\volav{ g(z)\, w T}.
\end{gather}
\end{subequations}
Observe that $\mathcal{F}\{T\}$ and $\mathcal{G}\{T\}$ are functionals of the temperature field only because $w$ is determined as a function of $T$ by solving~\eqref{e:w_and_T_eq} subject to the boundary conditions in~\eqref{e:w-bcs}. We shall prove that $\mathcal{F}\{T\}$ and $\mathcal{G}\{T\}$ are individually non-negative for all temperatures $T$ from the space $\Tspace_3$, which is sufficient for the spectral constraint to hold.  

To prove that $\mathcal{F}\{T\} \geq 0$, we apply \cref{lemma_3_main} with $\varphi(z) = f(z)/\bp $ and $\mu = A / \bp$, where $f$ is given by \eqref{e:j} and $A$ given by \eqref{e:A_choice_ih3}. We therefore need to check that
 \begin{equation} \label{eq:f_ww_cond}
 \volav{f(z)z^{1+\frac{\alpha}{2}}}^2 \leq \frac{\beta A}{R} \frac{\sqrt{2}(3-\alpha)(2+\alpha)\sqrt{1-\alpha^2}}{3+\alpha}.
 \end{equation}
 To verify this inequality, we first bound from above the weighted integral on the left-hand side. By assumption we have $0\leq z \leq \delta \leq \frac{1}{6}$ and $\alpha \in (\frac12, 1)$, from which we obtain  $\frac{1}{1-z} \leq  \frac{1}{z^{\alpha}}$. Using this estimate and the definition of $A$ from \eqref{e:A_choice_ih3} we can therefore estimate
\[
    \volav{f(z)z^{1+\frac{\alpha}{2}}} \leq   \int^{\delta}_0 z^{1+\frac{\alpha}{2}} + 2A z^{1-\frac{\alpha}{2}}\textrm{d}z =  \left[\frac{2}{4+\alpha} + \frac{8\sqrt{3}}{9} \frac{ \sqrt{2\alpha-1}}{(4-\alpha)}\delta^{\frac32 }\right] \delta^{2+\frac{\alpha}{2}} .
\]    
Using again that $0 \leq \delta \leq \frac16$ and $\alpha \in (\frac12,1)$, the bracketed expression can be bounded from above to obtain
\begin{equation} \label{e:est_integral_f}
     \volav{f(z)z^{1+\frac{\alpha}{2}}} \leq  \left[ \frac{4}{9} + \frac{8 \sqrt{3}}{9} \cdot \frac{1}{3} \cdot \left(\frac16\right)^\frac32 \right] \delta^{2+\frac{\alpha}{2}} < \frac12 \delta^{2 + \frac{\alpha}{2}}. 
\end{equation}

Next, we bound from below the right-hand-side of \eqref{eq:f_ww_cond}. Using the lower bound on $\bp$ from \cref{lemma:b_ih3}, the definition  \eqref{e:A_choice_ih3} of $A$, and the fact that $\alpha \in (\frac12,1)$, we have 
\begin{align}
\frac{\beta A \sqrt{2} (3-\alpha)(2+\alpha)}{(3+\alpha)} &\geq \frac{4 \sqrt{3}}{27} \delta^{\alpha + \frac{7}{2}} \sqrt{2\alpha-1} \cdot \frac{(3-\alpha)(2+\alpha)}{(3+\alpha)} \nonumber \\
& \geq \frac{5}{9\sqrt{3}} \delta^{\alpha + \frac72} \sqrt{2\alpha-1}. \label{eq:spec_lb}
\end{align}
Combining \eqref{e:est_integral_f} and \eqref{eq:spec_lb}, we conclude that \eqref{eq:f_ww_cond}  holds if  
\[
\delta^{\frac12} \leq \frac{20}{9 \sqrt{3}} \frac{\sqrt{(2\alpha-1)(1-\alpha^2)}}{R} \leq \sqrt{2} \frac{\sqrt{(2\alpha-1)(1-\alpha^2)}}{R} ,
\]
which is true because $\delta$ satisfies \subeqref{e:delta_eps_lemma}{a} by assumption. This proves that $\mathcal{F}\{T\} \geq 0$, as desired.

We now prove that $\mathcal{G}\{T\}$ is also nonnegative. This can be done following the same steps used in \cref{ss:lemmas-IH1}. The Hardy--Rellich inequality \subeqref{doering_estimate}{b} gives 
\begin{equation}
    \mathcal{G}\{T\} \geq \frac{\bp}{2} \volav{|\nabla T|^2} + \frac{4A}{R} \volav{ \frac{w^2}{(1-z)^3}} - A \volav{g(z) wT}.
    \label{e:T_upper_ih3}
\end{equation}
To estimate the last term, as before we use the inequality $\volav{T^2}_h\leq \smash{(1-z)\volav{ |\nabla T|^2} }$, the Cauchy--Schwarz inequality, and Young's inequality:
\begin{align*}
        \volav{g(z) w T} 
        \leq \frac{\bp}{2A} \volav{ |\nabla T|^2}  + \frac{A}{2\bp} \volav{ g(z)^2 (1-z)^4 } \volav{ \frac{w^2}{(1-z)^3} } .    
\end{align*}
Using the lower bound on $\bp$ from \eqref{e:b_lower_upper_ih3} gives
\begin{equation*}
     \volav{g(z) w T}  \leq \frac{\bp}{2A} \volav{ |\nabla T|^2}  + \frac{3\sqrt{2}\, A}{4 \delta^2}\volav{ g(z)^2 (1-z)^4 } \volav{ \frac{w^2}{(1-z)^3} } ,
\end{equation*}
which can be substituted into \eqref{e:T_upper_ih3} along with the value of $A$ from \eqref{e:A_choice_ih3} to obtain
\begin{equation} \label{e:T_exp_ih3}
    \mathcal{G}\{T\} \geq  \frac{\sqrt{6(2\alpha - 1)}\, A}{6} \left( \frac{4\sqrt{6}}{\sqrt{2\alpha-1}\, R} -\delta^{\alpha - \frac12} \volav{ g(z)^2 (1-z)^4 }\right)\volav{ \frac{|w|^2}{(1-z)^3}} .
\end{equation}
To conclude the argument we need to show that term in the parentheses is non-negative. To demonstrate this, we first estimate from above the function $g(z)$ given in \eqref{e:k} on the interval $(1-\varepsilon,1)$. 
Our assumption that $\delta\leq \varepsilon$ implies that $1-\varepsilon\leq 1-\delta \leq1$ and $\ln(\frac{1}{\varepsilon})\leq \ln(\frac{1}{\delta})$. Thus, for all $\delta\leq\frac16$ and $\alpha \in (\frac12,1)$ the first term in $g(z)$ can be bounded as
\begin{align}\label{e:esti_g_int}
    \frac{B}{\varepsilon} \leq \frac{1}{\varepsilon}\left(\frac{1}{1-\alpha} + \ln{\left(\frac{1}{\delta} \right)}\right)
    &\leq \frac{2\ln{\left(\frac{1}{\delta}\right)}}{\varepsilon(1-\alpha)} . 
\end{align}
To estimate the other terms in $g(z)$, we observe that the assumptions $\varepsilon\leq \frac13$ and $\alpha\in(\frac12,1)$ imply that  $\frac{1}{z^{\alpha}} \leq \frac{1}{2(1-\alpha)(1-z)}$ and  $\frac{1}{1-z} \leq \frac{1}{2(1-\alpha)(1-z)}$. Consequently, we arrive at
\begin{equation}\label{e:g_1st_est}
    g(z) \leq \frac{2\ln{(\frac{1}{\delta})}}{\varepsilon(1-\alpha)} + \frac{1}{(1-\alpha)(1-z)}.
\end{equation}
Finally, using \eqref{e:g_1st_est} and evaluating the integral in the parentheses of \eqref{e:T_exp_ih3} with the fact that $\ln(\frac{1}{\delta}) \leq  \ln^{2}{(\frac{1}{\delta})}$ and $\frac13 \leq \frac13 \ln{(\frac{1}{\delta})}$ gives
\begin{eqnarray}
    \delta^{\alpha - \frac12} \volav{g(z)^2(1-z)^4} &\leq& \delta^{\alpha - \frac12} \int^{1}_{1-\varepsilon} \left(\frac{2
    \ln{\left(\frac{1}{\delta}\right)}}{\varepsilon(1-\alpha)} + \frac{1}{(1-\alpha)(1-z)} \right)^2(1-z)^4\,\textrm{d}z  \nonumber\\
    &=& \frac{\delta^{\alpha - \frac12} \varepsilon^{3}}{(1-\alpha)^2} \left[\frac45 \ln^{2}{\left(\frac{1}{\delta}\right)} + \ln{\left(\frac{1}{\delta}\right)} + \frac13 \right]  \nonumber\\
    &\leq&   \frac{3\,\delta^{\alpha - \frac12}\varepsilon^{3} \ln^{2}{\left(\frac{1}{\delta}\right)}}{(1-\alpha)^2} \nonumber \\ {(\text{by \subeqref{e:delta_eps_lemma}{b}})} \qquad
    &\leq& \frac{4\sqrt{6}}{\sqrt{2\alpha - 1}\, R}  .
    \label{e:Phi_integral_est}
\end{eqnarray}
This concludes the proof of \cref{lemma:spectral_ih3}.

\end{proof}

\subsection{Proof of \texorpdfstring{\cref{thm:ih3}}{Theorem 2} } 
\label{ss:analtic_bound_ih3}
To prove \cref{thm:ih3}, we only need to specify $R$-dependent values for $\alpha$ and for the boundary layer widths $\delta$ and $\varepsilon$ that satisfy the conditions of \cref{lemma:b_ih3,,lemma:U_ih3,,lemma:spectral_ih3}. 

\begin{figure}
    \centering
    \includegraphics[scale=1]{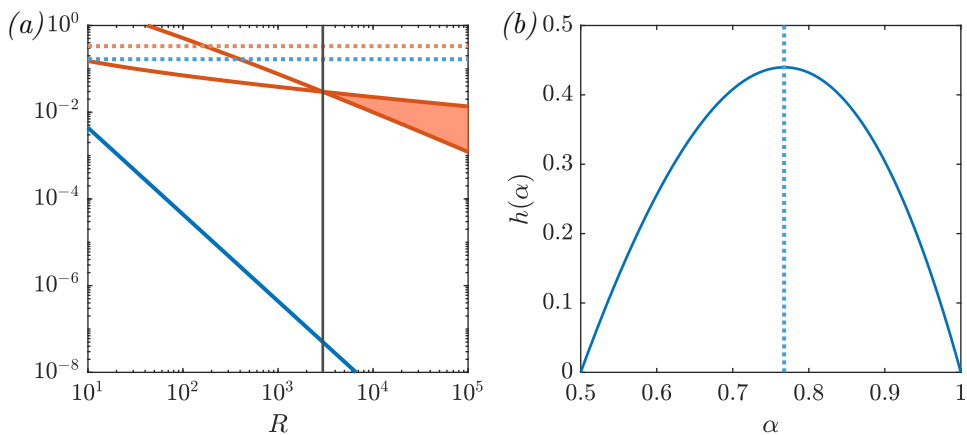}
    \begin{tikzpicture}[overlay]
        \node at (-12.5,5.4) {\textit{(a)}};
        \node at (-6,5.4) {\textit{(b)}};
    \end{tikzpicture}
    \caption{\textit{(a)} Variation with \Ra\ of the allowed values for the bottom boundary layer width $\delta$ \eqref{e:boundary_layers_ih3} ({\color{matlabblue}\solidrule}) and the feasible region of $\varepsilon$ \eqref{e:eps_range_alpha} (shaded region). Also shown are uniform upper bounds of $\delta\leq \frac{1}{6}$ ({\color{matlabblue}\dottedrule}), and $\varepsilon\leq\frac13$ ({\color{matlabred}\dottedrule}) imposed on the variables. A black vertical line marks the Rayleigh number, $\Ra_0\approx 2960.89 $ above which all constraints on are satisfied. \textit{(b)} Plot of the function $h(\alpha)$  \eqref{e:h_a_lem} ({\color{matlabblue}\solidrule}).
    Shown also is the optimal $\alpha^* = (1+\sqrt{13})/6$ ({\color{matlabblue}\dottedrule}).  }
    \label{fig:ih3_conditions}
\end{figure}

Motivated by the desire to minimize the upper bound on $U(\bp,\bfield,\lm)$ stated in \cref{lemma:U_ih3}, we choose 
\begin{equation}
    \delta =  h(\alpha^*)  R^{-2}
    \label{e:boundary_layers_ih3}
\end{equation}
where $\alpha^* = (1+\sqrt{13})/6$ is the unique maximizer of $h(\alpha)$ on the interval $(\frac12,1)$ (see \cref{fig:ih3_conditions}(b)). 
With these choices, conditions \subeqref{e:eps_bound_conditions}{c} and \subeqref{e:delta_eps_lemma}{b} require $\varepsilon$ to satisfy
\begin{eqnarray}
    c_0 R^{-4\alpha^* +2} B(\delta,\varepsilon,\alpha^*)^2 \leq \varepsilon \leq c_2 R^{-\frac{2-2\alpha^*}{3}} \ln^{-\frac23}{(\delta^{-1})},
\end{eqnarray}
where $c_0$, $c_1$ and $c_2$ are non-negative constants independent of \Ra. Using the upper bound on $B$ from~\eqref{e:esti_g_int}, it suffices to find $\varepsilon$ such that
\begin{eqnarray}
     \frac{4c_0}{(1-\alpha^*)^2} R^{-4\alpha^* +2}  \ln^2{\left(\frac{R^2}{h(\alpha^*)}\right)} 
     \leq \varepsilon 
     \leq c_2 R^{-\frac{2-2\alpha^*}{3}} \ln^{-\frac23}\left( \frac{R^2}{h(\alpha^*)} \right).
    \label{e:eps_range_alpha}
\end{eqnarray}
\Cref{fig:ih3_conditions} shows that suitable values of $\varepsilon$ exist when $\Ra \geq \Ra_0 \approx 2960.89$. One can also check that for all such values of \Ra\ and any $\varepsilon$ in the range given by \cref{e:eps_range_alpha} one has $\delta \leq \frac16$,  $\varepsilon \leq \frac13$, and $\delta \leq \varepsilon$. We have therefore verified all conditions of \cref{lemma:b_ih3,,lemma:U_ih3,,lemma:spectral_ih3}.

To conclude the proof of \cref{thm:ih3}, we simply substitute our choice of $\delta$ from \eqref{e:boundary_layers_ih3} into \cref{lemma:U_ih3} to find the upper bound $\mean{wT} \leq U(\bp,\bfield,\lm) \leq \frac12 - c R^{-4}$, 
where $c = \frac{h^2(\alpha^*)}{18}  \simeq 0.0107$.

\begin{remark}
 The top boundary layer width $\varepsilon$ is not uniquely determined in our construction. Its optimal value could be obtained by considering more refined estimates on $U(\bp,\bfield,\lm)$ than \cref{lemma:U_ih3}, but we expect such estimates to provide only higher-order corrections to the eventual bound on \wT.
\end{remark}

\section{Conclusions}\label{sec:conclusion}

We have proven upper bounds on the mean vertical convective heat transport \wT\ for two configurations of infinite-Prandtl-number convection driven by uniform internal heating between no-slip boundaries. In the first case, where both boundaries are held at a constant temperature, we find $\wT \leq \frac12 - O(\Ra^{-2})$ for all sufficiently large \Ra\ (cf. \cref{thm:ih1-main-result}). This result implies that the outward heat fluxes through the top and bottom are bounded by $\mathcal{F}_T \leq 1 - O(\Ra^{-2})$ and $\mathcal{F}_B \geq O(\Ra^{-2})$, respectively. In the second configuration, where the top boundary remains isothermal but the bottom one is insulating (no-flux condition), we find $\wT \leq \frac12 - O(\Ra^{-4})$ (cf. \cref{thm:ih3}). In this case, we conclude from \eqref{e:Nu_ih3} that the Nusselt number is bounded above by $\Nu \leq O(\Ra^4)$. Explicit suboptimal values for the prefactors in the Rayleigh-dependent terms were also obtained (cf. \cref{rem:prefactors-ih1,rem:prefactors-ih3}).

All of these results were derived using the background method, which we formulated as a search over quadratic auxiliary functionals of the form \eqref{e:V-IH1} and augmented using a minimum principle for the fluid's temperature. Similar to previous works on infinite-Prandtl-number Rayleigh--B\'enard convection, the background temperature fields used vary linearly in thin boundary layers, and increase either logarithmically (IH1 configuration) or as a power law (IH3 configuration) in the bulk of the fluid layer. This bulk behaviour enables us to use Hardy--Rellich inequalities  from~\cite{whitehead2011internal} (\cref{lemma:doering_estimate}) and an integral estimate from~\cite{Whitehead2014mixed} (\cref{lemma_3_main}) that were originally developed in the context of Rayleigh--B\'enard convection. In contrast to the latter, however, our background fields lack symmetry in the vertical direction, which reflects the lack of vertical symmetry of IH convection problems.

In our choice of background fields, allowing the bottom boundary layer width $\delta$ to be smaller than the top boundary layer width $\varepsilon$ is essential to prove \cref{thm:ih1-main-result,thm:ih3}.
For the IH1 configuration, forcing $\delta = \varepsilon$ worsens the \Ra-dependent correction to $\frac12$ in \cref{thm:ih1-main-result} to $O(\Ra^{-2}\ln^{-2}{(\Ra)})$. For the IH3 configuration, instead, no upper bound on \wT\ that asymptotes to $\frac12$ from below as \Ra\ increases can be obtained with our method of proof if $\delta=\varepsilon$. This boundary layer asymmetry contrasts the construction of background fields for IH convection at finite \Pr~\cite{kumar2021ihc}, where taking $\delta \neq \varepsilon$ appears to bring no qualitative improvement to the exponentially-varying upper bounds on \wT\ . 
%
We also stress that the \textit{a priori} uniform limits on the allowed values of $\delta$ and $\varepsilon$ imposed throughout \cref{sec:IH1,sec:IH3} have been chosen with the only goal of simplifying the algebra in our proofs. Varying these limits affects the prefactors of the \Ra-dependent terms in \cref{thm:ih1-main-result,thm:ih3}, as well as the range of \Ra\ values for which they hold. Both could be optimized further if desired.

One crucial difference between our constructions for the IH1 and IH3 configurations is the leading-order behaviour of the background temperature fields---or, more precisely, of the function $\bfield(z)$---as the bottom boundary layer edge is approached from the bulk region. 
For the IH1 configuration, it suffices for $\bfield$ to have the same logarithmic behaviour as the background temperature fields used to study Rayleigh--B\'enard convection~\cite{doering2006bounds}. For the IH3 configuration, however, this choice does not work due to the loss of control on the temperature of the bottom boundary, and we are instead forced to take $\bfield(z) \sim z^{1-\alpha}$ with $\alpha \in (0,1)$. This modification was already used in the context of Rayleigh--B\'enard convection between imperfectly conducting boundaries~\cite{Whitehead2014mixed}, where the optimal exponent $\alpha$ depended logarithmically on the Rayleigh number. Within our proof, instead, the optimal $\alpha$ is a constant. Whether this difference is due to our choice of estimates or the inherent differences between Rayleigh-B\'enard and IH convection remains an open question.

More generally, we do not know whether the upper bounds on \wT\ stated in \cref{thm:ih1-main-result,thm:ih3} are qualitatively sharp. To check if the $O(\Ra^{-2})$ and $O(\Ra^{-4})$ corrections to the asymptotic value of $\frac12$ are optimal within our bounding framework, one could employ a variation of the computational approach taken in~\cite{arslan2021IH1} and optimize the tunable parameters $\bfield$, $\bp$, and $\lm$ in full (see also~\cite{Fantuzzi2022} and references therein for more details on the numerical optimization of bounds). A more interesting but also more challenging problem is to identify which convective flows maximize \wT\, and the corresponding optimal scaling of this quantity with \Ra. Considerable insight in this direction can be gained through (i) direct numerical simulations, which to the best of our knowledge are currently lacking; (ii) the calculation of steady but unstable solution of the Boussinesq equations~\eqref{e:governing-equations} that, as recently observed in the context of Rayleigh--B\'enard convection~\cite{Wen2020,wen_goluskin_doering_2022}, may transport heat more efficiently than turbulence; and (iii) the explicit design of optimally-cooling flows~\cite{Tobasco2017,doering2019optimal,Tobasco2022}. Finally, it would be interesting to investigate if more sophisticated PDE analysis techniques used for Rayleigh--B\'enard convection~\cite{choffrut2016upper} can be extended to IH convection to interpolate between the algebraic bounds on \wT\ proved in this paper for infinite-\Pr\ fluids with the finite-\Pr\ exponential bounds obtained in~\cite{kumar2021ihc}.



\vspace{2ex}\noindent
\textbf{Acknowledgements} 
A.A. acknowledges funding by the EPSRC Centre for Doctoral Training in Fluid Dynamics across Scales (award number EP/L016230/1). G.F. was supported by an Imperial College Research Fellowship and would like to thank the Isaac Newton Institute for Mathematical Sciences, Cambridge, for support and hospitality during the programme ``Mathematical aspects of turbulence: where do we stand?'' (EPSRC grant number EP/R014604/1) where work on this paper was undertaken.



%% file: figs/tex/ih1-profiles-v2.tex
\begin{tikzpicture}[every node/.style={scale=0.9}, scale=0.75]
	\draw[->,black,thick] (-12.25,0.5) -- (-5,0.5) node [anchor=north] {$z$};
	\draw[->,black,thick] (-12,0.4) -- (-12,3.8) node [anchor=south] {$\bfield(z)$};
	\draw[matlabblue,very thick] (-12,3.3) -- (-11.2,0.5) ;
	\draw[matlabblue,very thick] (-6.3,2) -- (-5.5,0.5);
	\draw [matlabblue,very thick] plot [smooth, tension = 1] coordinates {(-11.2,0.5) (-9.75,1.2) (-7.25,1.5) (-6.3,2)};
	\node[anchor=north] at (-5.5,0.5) {$1$};
	\node[anchor=east] at (-12,3.2) {$1$};
	\node[anchor=north] at (-11.2,0.5) {$\delta$};
	\draw[dashed] (-6.3,0.5) node[anchor=north] {$1-\varepsilon$} -- (-6.3,2);
	\draw[->,black,thick] (-3.25,3.4) -- (4,3.4) node [anchor=west] {$z$};
	\draw[->,black,thick] (-2.5,0.4) -- (-2.5,3.8) node [anchor=south] {$\lm(z)$};
	\draw[colorbar12,very thick] (-2.5,0.5) -- (-1.5,0.5);
	\draw[colorbar12,very thick] (-1.5,3.4) -- (3.6,3.4) ;
	\draw[colorbar12,dotted] (-1.5,3.4) -- (-1.5,0.5) ;
	\node[anchor=south] at (3.6,3.4) {$1$};
	\node[anchor=south] at (-1.5,3.4) {$\delta$};
	\node[anchor=east] at (-2.5,0.5) {$-1/\delta$};
\end{tikzpicture}

%% file: figs/tex/ih3-profiles-v2.tex
\begin{tikzpicture}[every node/.style={scale=0.9}, scale=0.75]
   	\draw[->,black,thick] (-12.25,0.5) -- (-5,0.5) node [anchor=north] {$z$};
	\draw[->,black,thick] (-12,0.4) -- (-12,3.8) node [anchor=south] {$\bfield(z)$};
	\draw[matlabblue,very thick] (-12,2.3) -- (-11.2,0.5) ;
	\draw[matlabblue,very thick] (-6.3,1.8) -- (-5.5,0.5);
	\draw [matlabblue,very thick] plot [smooth, tension = 1] coordinates {(-11.2,0.5) (-9.75,1.2) (-7.25,1.5) (-6.3,1.8)};
	\node[anchor=north] at (-5.5,0.5) {$1$};
	\node[anchor=east] at (-12,2.3) {$\delta$};
	\node[anchor=north] at (-11.2,0.5) {$\delta$};
	\draw[dashed] (-6.3,0.5) node[anchor=north] {$1-\varepsilon$} -- (-6.3,1.8);
	\draw[->,black,thick] (-3.25,3.4) -- (4,3.4) node [anchor=west] {$z$};
	\draw[->,black,thick] (-2.5,0.4) -- (-2.5,3.8) node [anchor=south] {$\lm(z)$};
	\draw[colorbar12,very thick] (-2.5,0.5) -- (-1.5,0.5);
	\draw[colorbar12,very thick] (-1.5,3.4) -- (3.6,3.4) ;
	\draw[colorbar12,dotted] (-1.5,3.4) -- (-1.5,0.5) ;
	\node[anchor=south] at (3.6,3.4) {$1$};
	\node[anchor=south] at (-1.5,3.4) {$\delta$};
	\node[anchor=east] at (-2.5,0.5) {$-1$};
\end{tikzpicture}

%% file: arXiv/arxiv-appendix.tex
\appendix

\section{Convex duality for the IH1 configuration}\label{ss:duality-ih1}

The equivalence between \eqref{e:primal-form-ih1} and \eqref{e:dual-form-ih1} follow from a relatively standard convex duality argument. It will be enough to show that
\begin{equation}\label{e:ih1-duality-inf-sup}
    \sup_{T \in \Tspace_+}
    -\Phi\{T\}
    =
    \inf_{ \substack{\lm \in L^2(0,1) \\ \lm \text{ \rm nondecreasing} \\ \langle \lm \rangle = -1}}
    \tfrac{1}{4\bp} \volav{  \abs{ \bfield' - \lm - \bp \left(z - \tfrac{1}{2}\right) }^2},
\end{equation}
where
\begin{equation*}
    \Phi\{T\} := \volav{
        \bp \abs{\nabla T}^2 - R \bfield' (\Delta^{-2}\Delta_h T) T 
        - (\bfield' - \bp z) \partial_z T
        }.
\end{equation*}

To establish this identity, we start by rewriting the maximization on the left-hand side as a minimization problem for the Legendre transform of $\Phi$. Recall that the Legendre transform of a functional $\Lambda:\Tspace \to \bR$ is the functional
\begin{equation*}
\Lambda^*\{\mu\} := \sup_{T \in \Tspace} \left( \mu\{T\} - \Lambda\{T\}\right),
\end{equation*}
which acts on the dual space $\Tspace^*$ of bounded linear functionals $T \mapsto \mu\{T\}$ on $\Tspace$. We shall write $\Tspace^*_+ \subset \Tspace^*$ for the subset of nonnegative bounded linear functionals on $\Tspace$, meaning that $\mu \in \Tspace^*_+$ if and only if $\mu \in \Tspace^*$ and $\mu\{T\} \geq 0$ for all $T \in \Tspace_+$.

\begin{lemma}\label{th:ih1-duality}
	If the pair $(\bp,\bfield)$ satisfies the spectral constraint,
	$\sup_{T \in \Tspace_+} -\Phi\{T\}
	= \inf_{\mu \in \Tspace^*_+} \Phi^*\{\mu \}.$
\end{lemma}
%
\begin{proof}
	Define a functional $\Psi$ from $\Tspace$ into $\bR \cup \{+\infty\}$ via
	\begin{equation*}
	\Psi\{T\} := \begin{cases}
	0 & T \in \Tspace_+,\\
	+\infty & \text{otherwise}.
	\end{cases}
	\end{equation*}
	Its Legendre transform is
	\begin{equation*}
	\Psi^*\{\mu\} 
	=
	\begin{cases}
	0 &\text{if } -\mu \in \Tspace^*_+,\\
	+\infty &\text{otherwise}.
	\end{cases}
	\end{equation*}
	We claim that
	\begin{align*}
	\sup_{T \in \Tspace_+} -\Phi\{T\}
	&= -\inf_{T \in \Tspace_+} \Phi\{T\} \\
	&= -\inf_{T \in \Tspace} \left(\Phi\{T\} + \Psi\{T\} \right) \\
	&= \inf_{\mu \in \Tspace^*} \left( \Phi^*\{-\mu \} + \Psi^*\{\mu\} \right)\\
	&= \inf_{\mu \in \Tspace^*_+} \Phi^*\{\mu \}.
	\end{align*}
	The first, second and fourth equalities are immediate consequences of the definitions of $\inf$, $\sup$, $\Psi$, and $\Psi^*$. The third one, instead, follows from the Fenchel--Rockafellar minmax theorem when $\Tspace$ is viewed as a Hilbert space with the inner product $\volav{\nabla T_1 \cdot \nabla T_2}$ and norm $\langle \abs{\nabla T}^2 \rangle$. To apply this theorem as stated in~\cite[Theorem~1.12]{Brezis2010}, we need to verify that the functionals $\Phi$ and $\Psi$ are convex, and that $\Phi$ is continuous (with respect to the norm on $\Tspace$) at some $T_0 \in {\rm dom}(\Phi) \cap {\rm dom}(\Psi) = \Tspace_+$.
 
	The convexity of $\Psi$ is obvious, while that of $\Phi$ follows from the assumption that the pair $(\bp,\bfield)$ satisfies the spectral constraint (cf. \cref{def:sc}). 
	To see this, write $\Phi\{T\} = \mathcal{Q}\{T,T\} - (\bfield' - \bp z) \partial_z T$ where $\mathcal{Q}$ is the bilinear form
	\begin{equation*}
	    \mathcal{Q}\{T_1,T_2\} := \volav{\beta \nabla T_1 \cdot \nabla T_2 - R \bfield' (\Delta^{-2}\Delta_h T_1) T_2 },
	\end{equation*}
	and observe that the spectral constraint ensures $\mathcal{Q}\{T,T\} \geq 0$ for all $T$ in the linear space $\Tspace$. Thus, for any $\lambda \in [0,1]$ we can set $\bar{\lambda} = 1-\lambda$ and estimate
	\begin{align*}
	    \Phi\{\lambda T_1 + \bar{\lambda} T_2\} 
	    &= \mathcal{Q}\{\lambda T_1 + \bar{\lambda} T_2, \lambda T_1 + \bar{\lambda} T_2\} 
	    - (\bfield' - \bp z) \partial_z [\lambda T_1 + \bar{\lambda} T_2]\\
	    &= \lambda\Phi\{T_1, T_1\}+\bar{\lambda}\Phi\{T_2,T_2\}- \lambda \bar{\lambda} \mathcal{Q}\{T_1-T_2, T_1-T_2\}\\
	    &\leq \lambda \Phi\{T_1, T_1\}+\bar{\lambda} \Phi\{T_2,T_2\},
	\end{align*}
	proving that $\Phi$ is convex.
	The continuity of $\Phi$ at any $T_0 \in \Tspace_+$ follows because $\Phi$ is continuous on the whole space $\Tspace$. Indeed, the terms $\smallvolav{\abs{\nabla T}^2}$ and $\smallvolav{(\bfield'-\bp z)\partial_z T}$ in the expression for $\Phi\{T\}$ are clearly continuous on $\Tspace$. To see that the remaining term is also continuous, it is enough to establish that $T_k \to T$ in $\Tspace$ implies $\varphi_k := \smash{\Delta^{-2}\Delta_h T_k} \to \smash{\Delta^{-2}\Delta_h T} =: \varphi$ in $\Tspace$. This can be shown by combining the Poincar\'e inequalities $\smallvolav{\varphi^2} \lesssim \smallvolav{\abs{\Delta \varphi}^2}$ and $\smallvolav{T^2} \lesssim \smallvolav{\abs{\nabla T}^2}$ with the estimate
	\begin{equation*}
	\smallvolav{ \abs{\Delta \varphi}^2 }
	= \abs{\smallvolav{ \varphi \Delta^2 \varphi }}
	= \abs{\smallvolav{ \varphi \Delta_h T }}
	= \abs{\smallvolav{ \Delta_h \varphi T } }
	\leq \smallvolav{ \abs{\Delta \varphi}^2 }^\frac12 \smallvolav{ \abs{T}^2 }^\frac12,
	\end{equation*}
	which implies $\smallvolav{\abs{\Delta \varphi}^2} \leq \smallvolav{T^2}$. This concludes the proof of \cref{th:ih1-duality}.
\end{proof}

Next, we prove that since $\Phi\{T\}$ is invariant under horizontal translations of the temperature field $T$, the minimization of its Legendre transform $\Phi^*$ can be restricted to functionals $\mu \in \Tspace^*_+$ that are translation invariant. Specifically, for any real numbers $r,s$ define the translation map $\translation{r,s}:\Tspace \to \Tspace$ and its adjoint $\translation{r,s}^*:\Tspace^* \to \Tspace^*$ via
\begin{align*}
\translation{r,s} T &:= T(x+r, y+s, z) \\
\translation{r,s}^* \mu &:= \mu \circ \translation{r,s}.
\end{align*}
The functional $\mu \in \Tspace^*$ is translation invariant if $\translation{r,s}^* \mu= \mu$ for all $r$ and $s$. 

\begin{lemma}\label{th:ih1-invariance}
	Suppose $\Phi:\Tspace \to \bR$ satisfies $\Phi\{\translation{r,s}T\} = \Phi\{T\}$ for all $r,s \in \bR$. Then,
	\begin{equation*}
	\inf_{\mu \in \Tspace^*} \Phi^*\{\mu \} = \inf_{\substack{\mu \in \Tspace^* \\ \mu \text{ \rm transl. inv.}}} \Phi^*\{\mu \}.
	\end{equation*}
\end{lemma}
\begin{proof}
	It suffices to show that for every $\nu \in \Tspace^*$ there exists a translation-invariant $\mu \in \Tspace^*$ such that $\Phi^*\{\mu\} = \Phi^*\{\nu\}$. Since every temperature field $T \in \Tspace$ is horizontally periodic, such a $\mu$ can be constructed simply by averaging the functionals $\translation{r,s}^*\nu$ over horizontal translations $r$ and $s$, i.e., by letting
	\begin{equation*}
	\mu(T) := \fint_{r,s} \translation{r,s}^*\nu\{T\} = \fint_{r,s} \nu\{\translation{r,s} T\} \qquad \forall T \in\Tspace.
	\end{equation*}
	To show that $\Phi^*\{\mu\} = \Phi^*\{\nu\}$, we establish the complementary inequalities $\Phi^*\{\mu\} \leq \Phi^*\{\nu\}$ and $\Phi^*\{\mu\} \geq \Phi^*\{\nu\}$. For the first one, use the translation invariance of $\Phi$ and the definition of the Legendre transform to estimate
	\begin{equation*}
	\nu\{\translation{r,s}T\} - \Phi\{T\} = \nu\{\translation{r,s}T\} - \Phi\{\translation{r,s} T\} \leq \Phi^*\{\nu\}
	\end{equation*}
	for all $T \in \Tspace$ and all horizontal shifts $r,s \in \bR$. Averaging over horizontal shifts shows that $\mu\{T\} - \Phi\{T\} \leq \Phi^*\{\nu\}$ for all $T \in \Tspace$, which implies $\Phi^*\{\mu\} \leq \Phi^*\{\nu\}$.
	
	To obtain the reverse inequality observe that, by definition of $\Phi^*$, for any $\varepsilon>0$ there exists $T_\varepsilon \in \Tspace$ such that $\nu\{T_\varepsilon\} - \Phi\{T_\varepsilon\} \geq \Phi^*\{\nu\} - \varepsilon$. Then, since $\Phi$ is translation invariant,
	\begin{align*}
	\Phi^*\{\nu\} - \varepsilon 
	&\leq \nu\{T_\varepsilon\} - \Phi\{T_\varepsilon\} \\
	&= (\translation{r,s}^*\nu)\{\translation{-r,-s} T_\varepsilon\} - \Phi\{\translation{-r,-s} T_\varepsilon\}\\
	&\leq \Phi^*\{\translation{r,s}^*\nu\}.
	\end{align*}
	Upon averaging this inequality over all horizontal shifts $r,s \in \bR$ and applying Jensen's inequality to $\Phi^*$, which is concave because it is the supremum of linear functions, we find
	\begin{equation*}
	\Phi^*(\nu) - \varepsilon \leq 
	\fint_{r,s} \Phi^*\left(\translation{r,s}^*\nu \right)
	\leq \Phi^*\!\left(\fint_{r,s} \translation{r,s}^*\nu \right)
	= \Phi^*(\mu).
	\end{equation*}
	Letting $\varepsilon \to 0$ yields $\Phi^*\{\mu\} \geq \Phi^*\{\nu\}$, as desired. \Cref{th:ih1-invariance} is therefore proved.
\end{proof}

To establish identity~\eqref{e:ih1-duality-inf-sup} we now need to show that its right-hand side coincides with the infimum of $\Phi^*$ over translation-invariant functionals $\mu \in \Tspace^*_+$. For this, we use a characterization of such $\mu$ established in~\cite[Appendix C]{arslan2021IH1}.

\begin{lemma}\label{th:ih1-representation}
    Let $\Tspace^*_+$ be the set of positive linear functionals on the temperature space $\Tspace$ defined in~\eqref{e:T-space-ih1}.
	If $\mu \in \Tspace^*_+$ is translation invariant, there exists a nondecreasing function $\lm \in L^2(0,1)$ with $\volav{\lm}=-1$ such that $\mu\{T\} = \volav{-{\lm}(z)\, \partial_z T}$.
\end{lemma}

Thanks to this representation, all that remains to do is to calculate
\begin{equation}\label{e:Psi-transform-initial}
\Phi^*\{\mu\} = \sup_{T \in \Tspace} \volav{
	(\bfield' - \lm + \bp z) \partial_z T
	- \bp \abs{\nabla T}^2 
	+ R \bfield' (\Delta^{-2}\Delta_h T) T
}.
\end{equation}
To solve this maximization problem, let $\eta(z) = \smallvolav{T}_h(z)$ be the horizontal mean of $T$ and set $\xi = T - \eta$. Since $\Delta_h \eta =0$ and $\smallvolav{\xi}_h(z) = 0$ by construction, we can therefore substitute $T=\eta+\xi$ in~\cref{e:Psi-transform-initial} and solve the equivalent problem
\begin{multline}\label{e:Psi-transform-expanded}
\Phi^*\{\mu\} = 
    \sup_{\substack{\eta=0 \text{ if }z \in\{0,1\} \\ \xi = 0 \text{ if }z \in\{0,1\} \\ \smallvolav{\xi}_h=0}} 
    \Big\{
    \volav{
    	(\bfield' - \lm + \bp z) \eta'
    	- \bp \abs{\eta'}^2 
    	- \bp \abs{\nabla \xi}^2 
    	+ R \bfield' (\Delta^{-2}\Delta_h \xi) \xi
    }
\\[-5ex]
    \volav{
        (\bfield' - \lm + \bp z) \partial_z \xi
        - 2\bp\,  \eta'\, \partial_z \xi
        + R \bfield' \eta \,(\Delta^{-2}\Delta_h \xi)
    }
    \Big\}.
\end{multline}
The boundary conditions on $\eta$ and $\xi$ follow from those on $T$. 
The three terms on the second line vanish identically because $\smallvolav{\xi}_h(z) = 0$ at all $z\in[0,1]$. To verify this claim, observe that
\begin{equation*}
    \volav{f(z) \partial_z \xi} = \int_0^1 f(z) \volav{\xi}_h'(z) \dz
    = 0
\end{equation*}
for any function $f(z)$ that depends only on the vertical direction. Similarly, one can show that
\begin{equation*}
    \volav{R \bfield' \eta\, (\Delta^{-2}\Delta_h \xi)} = 0
\end{equation*}
because the function $w = -R\smash{\Delta^{-2}\Delta_h T} = -R\smash{\Delta^{-2}\Delta_h \xi}$ also has zero horizontal mean. Indeed, taking the horizontal average of~\eqref{e:w_and_T_eq} yields the ODE $\smallvolav{w}_h''(z)=0$, whose only solution satisfying the boundary conditions in~\eqref{e:w-bcs} is $\smallvolav{w}_h(z)=0$. The minimization in~\eqref{e:Psi-transform-expanded} therefore simplifies into
\begin{align}\label{e:Psi-transform-middle}
\Phi^*\{\mu\} 
= 
\sup_{\substack{\eta=0 \text{ if }z \in\{0,1\} \\ \xi = 0 \text{ if }z \in\{0,1\} \\ \smallvolav{\xi}_h=0}} 
\volav{
	(\bfield' - \lm - \bp z) \eta'
	- \bp \abs{\eta'}^2
	- \bp \abs{\nabla \xi}^2 
	- R \bfield' w \xi
}.
\end{align}
Since the pair $(\bp,\bfield)$ was assumed to satisfy the spectral constraint (cf. \cref{def:sc}), the choice $\xi=0$ is optimal. The optimal $\eta$, instead, satisfies the Euler--Lagrange equation
$2\bp \eta'' = (\bfield' - \lm - \bp z)'$. Solving this equation using the boundary conditions, the constraints $\bfield(0)=1$ and $\bfield(1)=0$, and the normalization $\smallvolav{\lm} = -1$ gives $\eta' = \frac{1}{2b}[\bfield' - \lm - b(z-\frac12)]$, which can be substituted back into~\eqref{e:Psi-transform-middle} to give
\begin{equation}
\Phi^*\{\mu\} = \tfrac{1}{4\bp} \langle \abs{ \bfield' - \lm' - \bp \left( z-\tfrac12 \right) }^2 \rangle.
\end{equation}
%
Minimizing the left-hand side over translation invariant $\mu$ in $\Tspace^*_+$ is the same as minimizing the right-hand side over $\lm$ satisfying the conditions in \cref{th:ih1-representation}, which is exactly the problem on the right-hand side of~\eqref{e:ih1-duality-inf-sup}.

\section{Convex duality for the IH3 configuration}\label{ss:duality-ih3}

The equivalence between the upper bounds~\cref{e:primal-form-ih3} and~\cref{e:dual-form-ih3} for the IH3 configuration follows from the identity
\begin{equation}\label{e:ih3-duality-inf-sup}
    \sup_{T \in \Tspace_+}
    -\Phi\{T\}
    =
    \inf_{ \substack{\lm \in L^2(0,1) \\ \lm \text{ \rm nondecreasing} \\ \lm \geq -1}}
    \volav{ \tfrac{1}{4\bp} \abs{ \bfield' - \lm - \bp z }^2 },
\end{equation}
where
\begin{equation}
    \Phi\{T\} := \volav{  
            \bp \abs{\nabla T}^2 
            -R \bfield' (\Delta^{-2}\Delta_h T) T  
            - \left( \bfield' - \bp z + 1 \right)\partial_z T
            }.
\end{equation}

This identity can be proven using a convex duality argument analogous to that in \cref{ss:duality-ih1}. Indeed, \cref{th:ih1-duality,th:ih1-invariance} apply to the functional $\Phi$ considered in this section with no changes to their proofs. Consequently,
\begin{equation}
    \sup_{T \in \Tspace_+} \label{e:ih3-duality-initial}
    -\Phi\{T\}
    =
    \inf_{\substack{\mu \in \Tspace^* \\ \mu \text{ \rm transl. inv.}}} \Phi^*\{\mu \}.
\end{equation}
To calculate the Legendre transform $\Phi^*$, however, we must replace \cref{th:ih1-representation} with a different characterization of translation-invariant linear functionals $\mu \in \Tspace^*_+$. This is due to the different boundary conditions imposed on the temperature space $\Tspace$.

\begin{lemma}\label{th:ih3-representation}
    Let $\Tspace^*_+$ be the set of positive linear functionals on the temperature space $\Tspace$ defined in~\eqref{e:T-space-ih1}.
	If $\mu \in \Tspace^*_+$ is translation invariant, there exists a nondecreasing function $\lm \in L^2(0,1)$ nonnegative almost everywhere and such that $\mu\{T\} = \volav{-{\lm}(z)\, \partial_z T}$.
\end{lemma}

\begin{proof}
    Straightforward modifications to the proof of~\cite[Lemma~3]{arslan2021IH1} reveal that any translation-invariant $\mu \in \Tspace^*$ admits the representation $\mu\{T\} = \volav{-{\lm}(z)\, \partial_z T}$ for some function $\lm \in L^2(0,1)$. If $\mu$ is also positive, then the argument in the proof  of~\cite[Lemma~2]{arslan2021IH1} shows that $\lm$ must be nondecreasing. To see that we must have $\lm(z_0) \geq 0$ at almost every $z_0 \in (0,1)$, fix $\varepsilon>0$ sufficiently small and consider the temperature field $T_\varepsilon \in \Tspace_+$ given by
    \begin{equation*}
        T_\varepsilon(\vx) := \begin{cases}
        2 &  z\in (0, z_0-\varepsilon),\\
        2-\varepsilon^{-1}(z-z_0+\varepsilon) & z \in (z_0-\varepsilon, z_0+\varepsilon),\\
        0 & z \in (z_0+\varepsilon, 1).
        \end{cases}
    \end{equation*}
    Then, since $\mu$ is a positive functional by assumption,
    \begin{equation*}
        0 \leq \mu(T_\varepsilon) = \frac1\varepsilon \int_{z_0-\varepsilon}^{z_0+\varepsilon} \lm(z) \dz. 
    \end{equation*}
    By Lebesgue's differentiation theorem, the right-hand side tends to $2\lm(z_0)$ for almost all $z_0 \in (0,1)$ as $\varepsilon \to 0$. Thus, we must have $\lm \geq 0$ almost everywhere on $(0,1)$.
\end{proof}

To conclude the argument, we need to calculate $ \Phi^*\{\mu \}$ for translation-invariant $\mu \in \Tspace^*_+$, which by \cref{th:ih3-representation} is given by
\begin{equation*}
\Phi^*\{\mu\} = \sup_{T \in \Tspace} 
\volav{
    \left( \bfield'  - \lm - \bp z + 1 \right)\partial_z T
    -\bp \abs{\nabla T}^2 
    +R \bfield' (\Delta^{-2}\Delta_h T) T
}.
\end{equation*}
Since the pair $(\bp,\bfield)$ was assumed to satisfy the spectral constraint, this maximization problem can be restricted to temperature fields that depend only on the vertical coordinate $z$ (this can be proven by splitting $T$ into its horizontal mean $\eta$ and a perturbation $\xi$ with zero horizontal mean, as outlined at the end of \cref{ss:duality-ih1}). The optimal value can then be shown to be 
\begin{equation*}
    \Phi^*\{\mu\} = \tfrac{1}{4\bp} \volav{ \abs{ \bfield' - \lm - \bp z + 1 }^2 }
\end{equation*}
and can be substituted into~\eqref{e:ih3-duality-initial} to arrive at
\begin{equation*}
    \sup_{T \in \Tspace_+} -\Phi\{T\}
    =
    \inf_{\substack{\lm \in L^2(0,1) \\ \lm \text{ \rm nondecreasing}\\\lm\geq 0}} 
    \tfrac{1}{4\bp} \volav{ \abs{ \bfield' - \lm - \bp z + 1 }^2 }.
\end{equation*}
Changing the optimization variable on the right-hand side to $\tilde{\lm} = \lm - 1$ yields~\eqref{e:ih3-duality-inf-sup}.